%% file: ms.tex
\documentclass[a4paper,UKenglish,cleveref, autoref]{lipics-v2019}

\nolinenumbers
\input{pcsptables}
\usepackage{mystyle}

\title{Dichotomy for symmetric Boolean PCSPs}

\author{Miron Ficak}{Theoretical Computer Science Department, Faculty of Mathematics and Computer Science, Jagiellonian
	University}{miron.ficak@student.uj.edu.pl}{https://orcid.org/0000-0003-3104-6354}{}
\author{Marcin Kozik}{Theoretical Computer Science Department, Faculty of Mathematics and Computer Science, Jagiellonian
	University}{marcin.kozik@uj.edu.pl}{https://orcid.org/0000-0002-1839-4824}{}
\author{Miroslav Ol\v s\'ak}{Department of Algebra, Charles University}{mirek@olsak.net}{}{}
\author{Szymon Stankiewicz}{Theoretical Computer Science Department, Faculty of Mathematics and Computer Science, JagiellonianUniversity}{szymon.stankiewicz@student.uj.edu.pl}{https://orcid.org/0000-0003-2235-4849}{}
\authorrunning{M. Ficak, M. Kozik, M. Ol\v s\'ak, S. Stankiewicz} 


\Copyright{Miron Ficak, Marcin Kozik, Miroslav Ol\v s\'ak, Szymon Stankiewicz} 

\funding{Research was partially supported by National Science Centre, Poland grant no. 2014/2013/B/ST6/01812.}
\relatedversion{A full version of the paper is avaialable at \url{TBA}.}
\category{Track A: Algorithms, Complexity and Games}
\EventEditors{Christel Baier, Ioannis Chatzigiannakis, Paola Flocchini, and Stefano Leonardi}
\EventNoEds{4}
\EventLongTitle{46th International Colloquium on Automata, Languages, and Programming (ICALP 2019)}
\EventShortTitle{ICALP 2019}
\EventAcronym{ICALP}
\EventYear{2019}
\EventDate{July 9--12, 2019}
\EventLocation{Patras, Greece}
\EventLogo{eatcs}
\SeriesVolume{132}
\ArticleNo{52}

\begin{document}
\hideLIPIcs

\maketitle

\begin{abstract}
  In one of the most actively studied version of Constraint Satisfaction Problem,
  a CSP is defined by a relational structure called a template.
  In the decision version of the problem 
  the goal is to determine whether a structure given on input admits a homomorphism into this template.
  Two recent independent results of Bulatov~[FOCS'17] and Zhuk~[FOCS'17] state that 
  each finite template defines CSP which is tractable or NP-complete.

  In a recent paper Brakensiek and Guruswami~[SODA'18] proposed an extension of the CSP framework.
  This extension, called Promise Constraint Satisfaction Problem,
  includes many naturally occurring computational questions, e.g. approximate coloring, that cannot be cast as CSPs. 
  A PCSP is a combination of two CSPs defined by two similar templates;
  the computational question is to distinguish a YES instance of the first one from  a NO instance of the second.

  The computational complexity of many PCSPs remains unknown.
  Even the case of Boolean templates~%
  (solved for CSP by Schaefer [STOC'78]) 
  remains wide open.
  The main result of Brakensiek and Guruswami [SODA'18] shows that Boolean PCSPs
  exhibit a dichotomy~(PTIME vs. NPC) when ``all the clauses are symmetric and allow for negation of variables''.
  In this paper we remove the ``allow for negation of variables'' assumption from the theorem.
  The ``symmetric'' assumption means that changing the order of variables in a constraint does not 
  change its satisfiability. 
  The ``negation of variables'' means that both of the templates share a relation which can be used
  to effectively negate Boolean variables.

  The main result of this paper establishes dichotomy for all the symmetric boolean templates. 
  The tractability case of our theorem and the theorem of Brakensiek and Guruswami are almost identical.
  The main difference, and the main contribution of this work, is the new reason for hardness
  and the reasoning proving the split.
\end{abstract}

\input{Introduction}

\input{Basic}

\input{TheoremAndP}
\input{OurNotation}

\input{Hardness}
\input{Sketch}

\appendix

\input{ManyOneSets}
\input{FewOneSets}


\bibliography{Bibliography}

\end{document}

%% file: Introduction.tex
\section{Introduction}
  Constraint Satisfaction Problems have been studied in computer science in many forms.
  In the general approach an instance of the CSP consists of variables and constraints.
  In the decision version of the problem the objective is to verify whether 
  there exists an evaluation of variables that meets all the constraints. 

  One particular type of CSPs received a lot of attention in the past years. 
  In this approach constraints are relations taken from a fixed, finite relational structure called a template.
  The interest in this particular version was driven by a conjecture of Feder and Vardi~\cite{FV98}
  postulating that each finite template defines a CSP which is tractable or NP-complete.

  A great variety of decision problems independently studied by computer scientists can be cast as CSPs.
  To name a few: 3-SAT, $k$-colorability, 
  (generalized) unreachability in directed graphs or solving systems of linear equation 
  over a finite field, are all CSPs defined by finite templates.
  The class of all the computational problems falling into the scope of the conjecture is very big
  and its verification was a gradual and lengthy process.
  Nevertheless, from the start, the claim was supported by strong evidence.
  In this context the classical result of Schaefer~\cite{Sch78} showing that the dichotomy holds for templates over Boolean domain,
  is perhaps the most important.

  The dichotomy for all the finite templates was recently confirmed by two, independent results of Bulatov~\cite{BulDich} and Zhuk~\cite{ZhukDich}.
  Both of them use the algebraic approach~\cite{JCG97,BJK00}, 
  where the complexity of a template is studied via compatible operations called polymorphisms.
  The algebraic approach proved very successful not only in the decision version of the CSP: 
  a number of important results in optimization~\cite{VCSP}, approximation~\cite{Robust} 
  etc. of the CSP is based on some versions of polymorphisms.

  A positive resolution of the dichotomy conjecture motivates the following question:
  is the class of CSPs unique, or maybe a part of a larger, natural class which also exhibits a dichotomy?
  Note that such a class should be amenable to some sort of the algebraic approach,
  as no other tools offer comparable power even in the case of the CSP.
  In the recent paper~\cite{PCSPSoda} Brakensiek and Guruswami proposed a candidate for such a class.

  The Constraint Satisfaction Problem defined by a fixed language
  can be cast as a problem of finding homomorphism from a relational structure given on input to a fixed template.
  The class proposed by Brakensiek and Guruswami as an extension of CSP is called Promise Constraint Satisfaction Problems.
  A PCSP is based on two CSPs with similar templates 
  and the question is to distinguish YES instances of the first CSP
  from NO instances of the second.

  To provide a few examples: 
  the CSP defined by an undirected clique~(without loops) of size $k$ as a template is just $k$-colorability.
  Defining PCSP by two cliques, say of sizes $k$ and $l$ satisfying $k<l$, we get the following problem:
  distinguish between the graphs with chromatic number $\leq k$ and those with chromatic number $>l$. 
  These problems are studied independently~\cite{DinurGraphs,Huang,BGGraphs,AlgebraicApproach}, 
  but the characterization of complexities for all pairs $(k,l)$ is either incomplete 
  or done under additional assumptions.

  Another example is a Boolean PCSP. 
  A single ternary relation $\set{(0,0,1),(0,1,0),(1,0,0)}$ defines a CSP which is known as Monotone-1-in-3-SAT,
  and similarly the relation $\set{0,1}^3\setminus\set{(0,0,0),(1,1,1)}$ gives rise to the CSP known as Monotone-NAE-SAT.
  Thus the question of distinguishing between instances which are satisfiable as Monotone-1-in-3-SAT instances
  and not satisfiable as Monotone-NAE-SAT instances is a PCSP.
  Surprisingly this problem is tractable even allowing for the negation of variables~\cite{TwoPlusE,PCSPSoda}.
  
  Further examples of problems expressible as PCSPs can be found in~\cite{PCSPSoda}.
  Promise Constraint Satisfaction Problems generalize CSPs 
  and include many additional, natural problems.
  The algebraic approach to the CSP can be adjusted to work in the case of the PCSP.
  The first Galois correspondence between PCSPs and the polymorphisms was introduced in~\cite{PCSPSoda},
  and the more abstract algebraic approach was proposed in~\cite{AlgebraicApproach}.
  Despite all the interest, PCSPs lack a classification result that would play the role of Schaefer's theorem.
  This motivates a more systematic study of Boolean PCSPs.
  
  The main result of Brakensiek and Guruswami,
  Theorem~2.1 in \cite{PCSPSoda},
  establishes dichotomy for a certain class of Boolean PCSPs.
  A PCSP template falls into this class if all the relations in the templates are symmetric~%
  (i.e. invariant under permutations, or equivalently, determined by Hamming weights of the tuples)
  and additionally the template contains a relation which can be used to negate Boolean variables in both CSP templates.
  As the additional relation is binary and symmetric,
  the result concerns all the symmetric templates containing this particular relation.
  In this paper we remove the additional assumption and show that all symmetric  Boolean templates exhibit a dichotomy.
  
  Let us further compare the results.
  The algorithms required for the original and extended result are exactly the same: 
  Gaussian elimination or linear programming relaxation depending on the polymorphisms of the template.
  The list of polymorphisms implying tractability differs slightly
  as we need to allow additional threshold functions~%
  (Boolean functions returning $0$ if and only if the number of $1$'s is below a threshold).
  Unfortunately the condition which guarantees hardness in the original paper fails 
  when the negating relation is absent.
  The new hardness condition and a more involved analysis of 
  the minion of polymorphisms are required in the proof 
  and constitute the main contribution of this paper.

  The publication is organized as follows.
  The next section contains basic definitions commonly used in context of an algebraic approach to the CSP or the PCSP.
  Section~\ref{sect:tractability} contains a list of  polymorphisms that guarantee tractability,
  statement of the main theorem and a proof of the tractability case.
  In section~\ref{sect:notation} we introduce notation and nomenclature.
  Section~\ref{sect:hardness} contains the algebraic condition implying hardness of PCSP
  and a proof of this implication.
  The main part of the reasoning behind the result 
  is focused on showing  
  that lack of 
  polymorphisms from the tractability list implies, in our case, the condition for hardness.
  Section~\ref{sect:proof} contains an overview of this proof
  and a complete reasoning can be found in with the majority of proper arguments  in the Appendix.

%% file: Basic.tex
\section{Basic definitions}\label{sect:basic}
This section contains basic definitions and notions
relevant to CSP and PCSP.
A relation $R\subseteq A^n$ is an {\em $n$-ary relation} 
and the set $A$ is its {\em universe}.
A  relation is {\em symmetric}, 
if for every permutation $\sigma$ of $\sset n$~%
(where $\sset n$ is defined to be $\set{1,\dotsc,n}$)
if $(a_1,\dotsc,a_n)\in R$ then also $(a_{\sigma(1)},\dotsc,a_{\sigma(n)})\in R$.
A relation $R^m\subseteq (A^m)^n$ is a {\em Cartesian power} of $R\subseteq A^n$  
if $(a^1,\dotsc,a^n)\in R^m$ if and only if $(a^1_i,\dotsc,a^n_i)\in R$ for every $i$~%
(i.e. $R^m$ is defined from $R$ coordinate-wise).

A {\em relational structure $\relstr A$} is a tuple $(A;R_1,\dotsc R_n)$ where each $R_i$ is a relation on $A$,
and we call a relational structure {\em symmetric} if all its relations are.
Two relational structures are \textit{similar} if they have the same sequence of arities of their relations.
E.g. a relational structure $(A;R_1,\dotsc R_n)$ and it's $m$-th power $(A^m;(R_1)^m,\dotsc (R_n)^m)$
are similar.
For two similar structures say $\relstr A = (A;R_1,\dotsc,R_n)$ and $\relstr B = (B;S_1,\dotsc,S_n)$ a 
function $h:A\rightarrow B$ is {\em a homomorphism} if for every $i$ and every tuple $(a_1,\dotsc,a_m)\in R_i$
the tuple $(h(a_1),\dotsc,h(a_m))\in S_i$.

The Constraint Satisfaction Problem defined by a relational structure $\relstr B$~%
(denoted by $\CSP(\relstr B)$) is the following decision problem:
\begin{center}
  \begin{tabular}{p{.15\textwidth}p{.8\textwidth}}
    Input: & a relational structure $\relstr A$ similar to $\relstr B$\\
    Question: & does there exists  a homomorphism from $\relstr A$ to $\relstr B$?
  \end{tabular}
\end{center}
The relational structure $\relstr B$ is called a {\em template} of such a problem.

The Promise Constraint Satisfaction Problem is a promise problem
defined by a pair of similar relational structures $(\relstr B, \relstr C)$~%
such that there exists a homomorphism from $\relstr B$ to $\relstr C$.
The $\PCSP(\relstr B,\relstr C)$ is:
\begin{center}
  \begin{tabular}{p{.15\textwidth}p{.8\textwidth}}
    Input: & a relational structure $\relstr A$ similar to $\relstr B$ and $\relstr C$ \\
    Output YES: & if there exists a homomorphism from $\relstr A$ to $\relstr B$ \\
    Output NO: & if there is no homomorphism from $\relstr A$ to $\relstr C$.
  \end{tabular}
\end{center}
Just like in the case of the CSP, the pair $(\relstr B,\relstr C)$ is called a template.
Clearly $\PCSP(\relstr B,\relstr B)$ is $\CSP(\relstr B)$ and therefore the PCSP generalizes the CSP.

Both problems exhibit a Galois correspondence i.e. instead of studying the structure of the template
one can choose to analyze the structure of template's polymorphisms~\cite{JCG97,BJK00,PCSPSoda,AlgebraicApproach}.
A {\em polymorphism} of a relational structure $\relstr B$ is a homomorphism from a finite Cartesian power of $\relstr B$
to $\relstr B$.
Similarly a polymorphism of a PCSP template $(\relstr B,\relstr C)$ 
is a homomorphism from a finite Cartesian power of $\relstr B$ to $\relstr C$.
We denote the set of all polymorphisms of $\relstr B$ by $\pol{\relstr B}$,
and the set of all polymorphisms of $(\relstr B,\relstr C)$ by $\pol{\relstr B,\relstr C}$.

For each relational structure $\relstr B$ the set $\pol{\relstr B}$ is {\em clone} i.e. 
it contains projections and is closed under composition.
Similarly for a pair $(\relstr B,\relstr C)$ the set $\pol{\relstr B,\relstr C}$ is a minion.
{\em A minion} is a set of functions closed under taking {\em minors} i.e. creating functions by identifying variables, permuting variables and introducing dummy variables.
If $f(x_1,\dotsc,x_n)$ is a function and $f'(x) = f(x,\dotsc,x)$
then $f'(x)$ is the unary minor of $f(x_1,\dotsc,x_n)$ 
and $f''(x,y) = f(x,y,\dotsc,y)$ is a binary minor of $f(x_1,\dotsc,x_n)$.

In some cases, instead of considering a PCSP template $\left((A;R_1,\dotsc,R_n),(B;S_1,\dotsc, S_n)\right)$
we work with an equivalent concept of {\em a language} i.e. a sequence of pairs $[R_1,S_1],\dotsc,[R_n,S_n]$.
We say that a pair $[S,T]$ is {\em compatible} with a minion $\clo M$,
if every member of $\clo M$ maps an appropriate power of $S$ to $T$~%
(the exponent of the power is the arity of the operation). 

{\em A primitive positive formula~(pp-formula)} is a formula constructed using
atomic formulas, conjunction and existential quantification.
Such formulas play a special role in CSP and PCSP:
if a relation $R$ has a primitive positive definition in $\relstr B$
then $R$ is compatible with $\pol{\relstr B}$ and adding $R$ to $\relstr B$ does not change the computational complexity of
the $\CSP(\relstr B)$.
Similarly, if a pair $[R,S]$ has a pp-definition in the language of $(\relstr A,\relstr B)$~%
(pp-formula in $[R_i,S_i]$ defines such $[R,S]$ in the natural way)
then $[R,S]$ is compatible with $\pol{\relstr A,\relstr B}$
and adding it to the language/template does not change the complexity \cite{PCSPSoda}.
One more construction,
called {\em strict relaxation},
plays an important role in the theory of PCSP:
if $[R_i,S_i]$ is an element of the language $(\relstr B,\relstr C)$
and $R\subseteq R_i$ while $S_i\subseteq S$ 
then $[R,S]$ is compatible with $\pol{\relstr A,\relstr B}$
and adding it to the language/template does not change the complexity.

%% file: TheoremAndP.tex
\section{Main theorem and tractability}\label{sect:tractability}
Focusing on the Boolean domain we present the main theorem of the paper 
and prove that the tractable cases are indeed solvable in P.
In this part of the proof our paper does not deviate much from~\cite{PCSPSoda};
the polymorphisms which imply tractability are almost the same with an exception
of the threshold case. 

\begin{itemize}
  \item A $n$-ary function is {\em a max}~({\em a min}) if it returns maximum~(resp. minimum) of its arguments~%
    (in the natural order on $\{0,1\}$).
  \item A function $f(x_1,\dotsc,x_n)$ is {\em an alternating threshold} if $n=2k+1$ and
  \[
    f(x_1,\ldots,x_k,x_{k+1},\ldots,x_n)=
    \begin{cases}
      0&\text{ if $\sum_{i=1}^k x_i \geq \sum_{i=k+1}^{n} x_i$,}\cr
      1&\text{ if $\sum_{i=1}^k x_i < \sum_{i=k+1}^{n} x_i$,}\cr
    \end{cases}
  \]
  \item A function $f(x_1,\dots,x_n)$ is {\em an xor} if $n$ is odd and
    $f(x_1,\dotsc,x_n)= x_1+\dotsb+x_n\bmod 2$.
  \item A function $f(x_1,\dotsc,x_n)$ is {\em a $q$-threshold}~%
    (where $q$ is a rational between $0$ and $1$) if
    \[
      f(x_1,\ldots,x_n)=
      \begin{cases}
        0&\text{ if $\sum_{i=1}^n x_i < nq$,}\cr
        1&\text{ if $\sum_{i=1}^n x_i > nq$,}\cr
      \end{cases}
    \]
    and $nq$ is not an integer.
    Note that all the evaluations of the $f(x_1,\dotsc,x_n)$ are determined.
\end{itemize}
We denote the set of all max functions by \Max,
all the min functions by \Min,
all alternating thresholds by \AT
all xor by \Xor and all $q$-thresholds by \THR{q}.
For a set of functions $F$ by $\compl F$ we denote $\{1-f(x_1,\dotsc,x_n)\suchthat f(x_1,\dotsc,x_n)\in F\}$.
We are ready to state the main result of the paper.

\begin{theorem}\label{thm:main}
  Let $(\relstr A,\relstr B)$ be a symmetric, Boolean PCSP language.
  If $\pol{\relstr A,\relstr B}$ contains a constant or includes at least one of the sets
  $\Max$, $\Min$, $\AT$, $\Xor$ , $\THR{q}$~(for some $q$),
  $\compl \Max$, $\compl \Min$, $\compl \AT$, $\compl \Xor$ or $\compl{\THR{q}}$~(for some $q$)
  then $\PCSP(\relstr A,\relstr B)$ is tractable.
  Otherwise it is NP-complete.
\end{theorem}

\noindent Comparing the statement of Theorem~\ref{thm:main} and Theorem~2.1 of~\cite{PCSPSoda}
we find two differences:
the earlier paper additionally assumes that negated variables can appear in instances
and it allows the authors to substitute ``$\THR{q}$ for some $q$'' with $\THR{1/2}$
in the list of conditions that force tractability.

In the remaining part of this section we will show the tractability case of Theorem~\ref{thm:main}.
The reasoning differs very little from the one found in~\cite{PCSPSoda} and 
therefore we cover it quickly:
If $\pol{\relstr A,\relstr B}$ contains a constant function $\PCSP(\relstr A,\relstr B)$ is clearly tractable;
if it includes $\Max$, $\Min$ and $\Xor$ tractability follows from Lemma~3.1 of~\cite{PCSPSoda}.
If $\AT\subseteq\pol{\relstr A,\relstr B}$ then Claim 2 of Section 3.2~\cite{PCSPSoda} implies tractability.
Finally the case of $\THR{q}$ is a minor generalization of the argument in Claim 1 of Section 3.2 in the same paper,
or a special case of Theorem 5.2 in~\cite{BlendLinearRings}.

The remaining cases reduce, just like in~\cite{PCSPSoda}, to the ones from the previous paragraph:
let relational structure $\relstr{B'}$ be obtained from $\relstr B$ by exchanging the roles of $0$ and $1$~%
(that is, in every relation in $\relstr B$, in every tuple of this relation and at every position in this tuple we change $x$ to $1-x$).
The YES instances of $\PCSP(\relstr A,\relstr{B'})$ and $\PCSP(\relstr A,\relstr B)$ are trivially the same
and so are the NO instances.
If $\compl\Min\subseteq\pol{\relstr A,\relstr B}$ then $\Min\subseteq\pol{\relstr A,\relstr{B'}}$
and, by the cases already established, $\PCSP(\relstr A,\relstr{B'})$ is tractable.
Clearly $\PCSP(\relstr A,\relstr B)$ is tractable as well and all the remaining tractable cases can be dealt with the same way.

%% file: OurNotation.tex
\section{The notation for symmetric Boolean PCSPs}\label{sect:notation}
  In order to show NP-hardness in the remaining case of Theorem~\ref{thm:main},
  we require a few definitions which allow us to work with symmetric Boolean relations
  and Boolean function concisely.

  Every symmetric relation 
  $R\subseteq \{0,1\}^m$  
  is uniquely determined by the set
  $I\subseteq \set{0,\dotsc,m}$ consisting of
  the Hamming weights of its elements.
  This fact allows us to use $R$ and $I$ interchangeably.
  Let $(\relstr B,\relstr C)$ be a symmetric, Boolean PCSP template with language $[R_1,S_1],\dotsc,[R_n,S_n]$
  where the arities of the relations are $a_1,\dotsc,a_n$.
  We will denote such a language by
  $\rel{I_1,J_1,a_1}, \dotsc,\rel{I_n,J_n,a_n}$ 
  where $I_i$~($J_i$) is a set of Hamming weights of  elements of $R_i$~($S_i$ respectively). 
  We will often use a flattened form of this notation: we will denote $\rel{\set{1},\set{1,2},2}$ by $\rel{{1},{1,2},2}$ and so on
  as well as 
  $\sset{n} = \set{1, \ldots, n}$.
  
  Focusing on compatibility;
  an operation $f(x_1,\dotsc,x_n)$ is compatible with $\rel{0,0,1}$ if and only if 
  $f(0,\dotsc,0)=0$ and compatible with $\rel{1,1,1}$ if and only if $f(1,\dotsc,1)=1$.
  The pair $\rel{1,1,2}$ defines negation in $\relstr A$ and $\relstr B$
  and therefore the main result of~\cite{PCSPSoda} is a special case of Theorem~\ref{thm:main};
  the additional assumption states that $\rel{1,1,2}$ is in the language of PCSP.

  We proceed to illustrate a number of pp-definitions and strict relaxations
  that appear repeatedly in the proofs.
  Using $\rel{I,J,n}$ and $\rel{0,0,1}$ one can define  
  $\rel{I\setminus \{n\},J\setminus\{n\},n-1}$ using the following pp-formula:
  \begin{equation*}
    \exists x_1\  \rel{0,0,1}(x_1)\wedge \rel{I,J,n}(x_1,\dotsc,x_n).
  \end{equation*}
  Similarly 
  \begin{equation*}
    \exists x_1\  \rel{1,1,1}(x_1)\wedge \rel{I,J,n}(x_1,\dotsc,x_n)
  \end{equation*}
  defines $\rel{I',J',n-1}$ where 
  $I' = \{i-1\suchthat i\in I\text{ and } i\neq 0\}$ and 
  $J' = \{j-1\suchthat j\in J\text{ and } j\neq 0\}$. 
  The strict relaxations we use are straightforward: 
  take $\rel{I, J, n}$ with $i\in I$ while $j\notin J$
  then, for example, $\rel{i,{\{0,\dotsc,n\} \setminus \set {j} },n}$ is a strict relaxation of $\rel{I,J,n}$.

  In the proof of tractability for $(\relstr B,\relstr C)$~%
  (at the end of Section~\ref{sect:tractability})
  we swapped the role of $0$ and $1$ in $\relstr C$.
  In the new notation we change
  $\rel{I_1,J_1,a_1}, \dotsc,\rel{I_n,J_n,a_n}$
  to 
  $\rel{I_1,J'_1,a_1}, \dotsc,\rel{I_n,J'_n,a_n}$ 
  where
  $J'_k = \{a_k - j\suchthat j\in J_k\}$.
  In some of the proofs we reuse this construction,
  although we usually swap for both $\relstr B$ and $\relstr C$ at the same time.

  We define notation for Boolean functions next.
  A Boolean function $f(x_1,\dotsc,x_n)$ is {\em idempotent}
  if $f(0,\dotsc,0) = 0$ and $f(1,\dotsc,1)=1$.
  By the discussion above a minion is idempotent~%
  (i.e. contains idempotent functions only)
  if it is compatible with $\rel{0,0,1}$ and $\rel{1,1,1}$.
  Moreover the idempotent part of $\pol{\relstr B,\relstr C}$
  can be obtained by adding these pairs to the language.

  For a Boolean function $f(x_1,\dotsc,x_n)$ and a set $U\subseteq\sset{n}$
  the value $f(U)$ is defined as $f(x_1,\dotsc,x_n)$ where $\set{i\suchthat x_i=1} = U$.
  When $n$ is clear from the context we can write $\compl U$ instead of $\sset{n}\setminus U$.
  Let $f(x_1,\dotsc,x_n)$ be a Boolean function $U\subseteq \sset{n}$ then $U$ is 
  \begin{itemize}
    \item a \oneset if $f(U) = 1$, 
    \item a \zeroset if $f(\compl U) = 0$,
    \item a \onefset(\zerofset) if every $V\supseteq U$ is a \oneset(resp. \zeroset).
  \end{itemize}

  Moreover we say that a minion has {\em small fixing sets},
  if there exists a constant $N$ such that every function from the minion has a \onefset smaller than $N$,
  or every function from the minion has a \zerofset smaller than $N$.
  Finally we say that a minion has {\em bounded antichains},
  if there exist a constant $M$ such that no function in the minion has $M$ pairwise disjoint \onesets,
  and no function in the minion has $M$ pairwise disjoint \zerosets.

%% file: Hardness.tex
\section{The hardness proof}\label{sect:hardness}
  
  In order to satisfy the assumptions of Lemma \ref{l:erobust}, 
  we need some structural properties of the minion $\pol{\relstr A,\relstr B}$.
  The following theorem collects these properties and is a cornerstone of our classification. 
  
  \begin{theorem}\label{thm:real}
    Let $\relstr A,\relstr B$ be a symmetric PCSP language such that $\pol{\relstr A,\relstr B}$ is idempotent.
    If $\pol{\relstr A,\relstr B}$ does not include $\Max$, $\Min$, $\AT$, $\Xor$ and $\THR{q}$~(for any $q$),
    then $\pol{\relstr A,\relstr B}$ has small fixing sets and bounded antichains.
  \end{theorem}
  The Brakensiek and Guruswami~\cite{PCSPSoda} version of Theorem~\ref{thm:real}
  requires that $(\relstr A,\relstr B)$ contains $\rel{1,1,2}$ and concludes 
  that there exists a constant $M$ 
  such that every member of $\pol{\relstr A,\relstr B}$ has a set of size at most $M$
  which is a \onefset and a \zerofset at the same time.
  The following example illustrates that their condition fails in our case.

  \begin{example}
    Consider PCSP defined by a language consisting of $\rel{0,0,1}$, $\rel{1,1,1}$, $\rel{1,{1,2},3}$ and $\rel{1,{1,2},4}$.
    It is easy to verify that it falls into the hardness case of Theorem~\ref{thm:main}.
    On the other hand for each odd $n$  the function $f(x_1,\dotsc,x_n)$ defined as maximum of $x_1$ and $n$-ary element of $\THR{1/2}$
    is compatible with all the relational pairs.
    These functions have no uniform bound on the size of minimal \zerofsets. 
  \end{example}

  In the reminder of this section we use Theorem~\ref{thm:real} to finish the proof of Theorem~\ref{thm:main}.
  We begin by introducing the machinery developed in~\cite{AlgebraicApproach}~%
  (a direct proof is possible, but involves a bit more technical considerations).
  The paper~\cite{AlgebraicApproach} defines {\em minor identity} as a 
  formal expression of the form
  \begin{equation*}
    f(x_1,\dotsc,x_n)\approx g(x_{\pi(1)},\dotsc,x_{\pi(m)})
  \end{equation*}
  where $f$ and $g$ are function symbols~(of arity $n$ and $m$, respectively), $x_1, \ldots, x_n$ are
  variables, and $\pi : \sset{m} \rightarrow \sset{n}$.
        A minor identity is satisfied in a minion $\clo M$~(of functions from $A$ to $B$) if there exists an interpretation
        of the function symbols $f$ and $g$ in $\clo M$, say $\zeta$, satisfying
  \begin{equation*}
  	\zeta (f) (a_1, \ldots, a_n) = \zeta(g)(a_{\pi(1)}, . . . , a_{\pi(m)})
  \end{equation*}
  	for all $a_1, \ldots, a_n \in A$.
  	
        {\em A bipartite minor condition} is a finite set of minor identities
  	in which function symbols used on the right- and left-hand sides are
  	disjoint. 
        A minor condition is {\em satisfied in a minion},
        if there exists an interpretation simultaneously satisfying all the identities. 
        A minor condition is {\em trivial} if it is satisfied in every minion, in particular,
  	in the minion consisting of all projections on a set A that
  	contains at least two elements.
        Finally, still following~\cite{AlgebraicApproach},
        a bipartite minor condition $\Sigma$ is $\varepsilon$-robust~(for some $\varepsilon > 0$)
  	if no $\varepsilon$-fraction of identities from $\Sigma$ is trivial.
  
  \begin{lemm}[Corollary~5.8 from \cite{AlgebraicApproach}]\label{l:erobust}
  	If there exists an  $\varepsilon > 0$ such that $\pol{\relstr A,\relstr B}$ does not satisfy any
  	$\varepsilon$-robust bipartite minor condition, then $PCSP(\relstr A,\relstr B)$ is NP-hard.
  \end{lemm}
In order to apply Lemma \ref{l:erobust} to $\PCSP(\relstr A,\relstr B)$ we need to
ensure that $\pol{\relstr A,\relstr B}$ does not satisfy any $\varepsilon$-robust bipartite minor condition.
Our first step is to prove it in the idempotent case.
  \begin{prop}\label{prop:epsi}
    Let $\clo M$ be an idempotent minion with small fixing sets,
    and bounded antichains.
    Then $\clo M$ does not satisfy any $\varepsilon$-robust bipartite minor condition.
  \end{prop}
  \begin{proof}
    The proof follows the same pattern as the proofs of Propositions 5.10 and 5.12 in~\cite{AlgebraicApproach} so we will use the notation from those Propositions in this proof.
    All we need to do is to find $\varepsilon>0$ 
    and a mapping assigning to each member of $\mathscr M$ 
    a probability distribution on its variables.
    The probability distribution needs to satisfy the following condition:
    if $f,g\in\mathscr M$ and $f(x_1,\dotsc,x_n)\approx g(x_{\pi(1)},\dotsc,x_{\pi(m)})$
    then  
    \begin{itemize}
      \item choosing a variable from the LHS according to the distribution for $f$ and
      \item choosing a variable from the RHS according to the distribution for $g$, 
    \end{itemize}
    with probability greater than $\varepsilon$ we will choose the same variable.

    In order to find such $\varepsilon$ and the mapping for $\mathscr M$ we assume
    without loss of generality that small fixing sets in $\mathscr M$ are \onefsets
    and their size as well as a size of an antichain is bounded by constant $M$.
    We choose $\varepsilon < 1/M^4$ and define the probability distribution as follows:
    fix $f\in\mathscr M$ and from the collection of \onefsets smaller than $M$ choose a maximal 
    subset of pairwise disjoint \onefsets.
    Let $U_f$ be the set of numbers appearing in this subset and
    the probability distribution for $f$ is the uniform distribution on $U_f$.

    Take an identity as above; as $|U_f|\leq M^2$ and $|U_g|\leq M^2$ in order 
    to prove the claim it suffices to show that $\pi(U_g)\cap U_f\neq\emptyset$.
    Let $U$ be one of the \onefsets which defined $U_g$.
    The set $\pi(U)$ is a \onefset of $f$ and its size is bounded by $M$.
    The maximality of the subset defining $U_f$ implies that $U_f$ and $\pi(U)$ intersect,
    which concludes the proof.
  \end{proof}

  We are now ready to finish the proof of Theorem~\ref{thm:main}~%
  (modulo Theorem~\ref{thm:real}) following a reasoning similar to the one used in~\cite{PCSPSoda}.
  Let $(\relstr B,\relstr C)$ be a PCSP language such that  $\pol{\relstr B,\relstr C}$ doesn't contain constant functions
  and do not include any of 
  $\Max$, $\Min$, $\AT$, $\Xor$, $\THR{q}$, $\compl \Max$, $\compl \Min$, $\compl \AT$, $\compl \Xor$, $\compl{\THR{q}}$.
  Let $(\relstr B_+,\relstr C_+)$ be $(\relstr B,\relstr C)$ with $\rel{1,1,1}$ and $\rel{0,0,1}$ added.
  By Theorem~\ref{thm:real} and Proposition~\ref{prop:epsi} 
  $\pol{\relstr B_+,\relstr C_+}$ does not satisfy any $\varepsilon$-robust minor condition~%
  (for some fixed $\varepsilon$).
  Note that $\pol{\relstr B_+,\relstr C_+}$ consists of these elements of $\pol{\relstr B,\relstr C}$
  which have identity as the unary minor.
  Thus $\pol{\relstr B,\relstr C}\setminus\pol{\relstr B_+,\relstr C_+}$ consists of elements of $\pol{\relstr B,\relstr C}$
  which have $x\mapsto 1-x$ as the unary minor.

  Consider the set $\compl{\pol{\relstr B,\relstr C}\setminus\pol{\relstr B_+,\relstr C_+}}$.
  It is a minion and it is equal to $\pol{\relstr B_-,\relstr C_-}$,
  where $(\relstr B_-,\relstr C_-)$ is obtained from $(\relstr B,\relstr C)$ in two steps:
  first the roles of $0$ and $1$ are swapped in $\relstr C$~%
  (just like in the tractability proof)
  and then $\rel{1,1,1}, \rel{0,0,1}$ are added to the language.
  Applying Proposition~\ref{prop:epsi} to $(\relstr B_-,\relstr C_-)$ 
  we conclude that $\compl{\pol{\relstr B,\relstr C}\setminus\pol{\relstr B_+,\relstr C_+}}$~%
  does not satisfy any $\varepsilon$-robust minor condition~%
  (for some $\varepsilon$).
  The same holds for $\pol{\relstr B,\relstr C}\setminus\pol{\relstr B_+,\relstr C_+}$ and
  therefore  $\pol{\relstr B,\relstr C}$ is a disjoint union of two minions which, for some $\varepsilon$, do not satisfy
  any $\varepsilon$-robust minor conditions.
  It follows that $\pol{\relstr B,\relstr C}$ does not satisfy any $\varepsilon$-robust minor condition
  and by Lemma \ref{l:erobust} the $\PCSP(\relstr B,\relstr C)$ is NP-hard.

%% file: Sketch.tex
\section{Proof overview}\label{sect:proof}
Our proof of Theorem~\ref{thm:real} consists of the following four propositions.
\begin{prop}
  \label{arel-crel}
  Let $(\relstr A,\relstr B)$ be a symmetric language such that $\clo M = \pol{\relstr A,\relstr B}$ is idempotent.
  If $\clo M$ does not include neither $\Max$ nor $\Min$,
  then it is compatible with some relational pair $\arel$
  and some relational pair $\crel$.
\end{prop}

\noindent For the next proposition we need to specialize the notion of bounded antichains.
We say that a minion has {\em bounded antichains of \onesets~(\zerosets)} if 
there exists a uniform bound on the number of pairwise disjoint
\onesets~(\zerosets respectively)
an element of the minion can have.

\begin{prop}
  \label{antichain-equivalent}
  Let $\clo M$ be a minion compatible with $\arel$ and $\crel$.
  Then $\clo M$ has bounded antichains of \onesets if and only if
  $\clo M$ has bounded antichains of \zerosets.
\end{prop}

\begin{prop}
  \label{get-xor-at}
  Let $(\relstr A,\relstr B)$ be a symmetric language such that $\clo M = \pol{\relstr A,\relstr B}$ is idempotent.
  If $\clo M$ is compatible with some $\arel$, some $\crel$ 
  and does not have bounded antichains then $\clo M$ includes $\Xor$ or $\AT$.
\end{prop}

\begin{prop}
  \label{get-small-fixing}
  Let $(\relstr A,\relstr B)$ be a symmetric language such that $\clo M = \pol{\relstr A,\relstr B}$ is idempotent.
  If $\clo M$ has bounded antichains
  and does not include any of $\THR q$ then it has
  small fixing sets.
\end{prop}

The structure of the proof is as follows:
if $\pol{\relstr A,\relstr B}$ has $\Min$ or $\Max$ we are in a tractable case.
Otherwise we split the reasoning in two cases:
either $\pol{\relstr A,\relstr B}$  fails the bounded antichain condition and by Proposition~\ref{get-xor-at} we are tractable due to $\AT$ or $\Xor$,
or we have bounded antichains and by Proposition~\ref{get-small-fixing} we are either tractable due to $\THR q$ or have small fixing sets
which implies hardness (by Proposition~\ref{prop:epsi}).
Proposition~\ref{antichain-equivalent} allows us to ``flip'' the template if necessary.

In this section, we prove Propositions~\ref{arel-crel}
and~\ref{antichain-equivalent}. We also provide proof sketches of
Propositions~\ref{get-xor-at} and~\ref{get-small-fixing}.
Detailed proofs can be found in Appendices~\ref{app:get-xor-at} and~\ref{app:get-small-fixing} respectively.

\begin{proof}[Proof of Proposition~\ref{arel-crel}]
  The proof splits into two parts:
  \begin{itemize}
  \item $\clo M$ does not have $\Min$ then $\clo M$ is compatible with $\arel$
  \item $\clo M$ does not have $\Max$ then $\clo M$ is compatible with $\crel$
  \end{itemize}
  Proof of both cases is analogous, so we will only prove the first
  part. Let us assume that $\clo M = \pol{\relstr A,\relstr B}$ and $\clo M$ does
  not have $\Min$. So there must be $\rel{I, J, n}$ in the language of $(\relstr A,\relstr B)$ such
  that $\Min$ is not compatible with it. This implies that there exists
  $b<a<n$ such that $a \in I$ and $b \not \in J$. 
  Now, using pp-definitions and strict relaxations from Section~\ref{sect:notation}, 
  we will show that $\clo M$ is compatible with $\rel{a-b,{1,\dotsc,a-b+1},a-b+1}$:
  \begin{itemize}
    \item use strict relaxation of $\rel{I, J, n}$ to obtain $\rel{a,{0,\dotsc,b-1,b+1,\dotsc, n} , n}$;
    \item from the last pair pp-define, using $\rel{0,0,1}$, the pair $\rel{a, {0,\dotsc,b-1,b+1,\dotsc,a+1}, a+1}$,
    \item finally from the previous pair pp-define, this time using $\rel{1,1,1}$, the required pair $\rel{a-b,{1,\dotsc,a-b+1},a-b+1}$.
  \end{itemize}
 
\end{proof}

\noindent The following lemma is used in the proof of Proposition~\ref{antichain-equivalent}.
\begin{lemm}\label{onesettozeroset}
	Let $\clo M$ be a minion. Then:
	\begin{itemize}
		\item if $\clo M$ is compatible with some $\arel$, 
                  then for each $f$ in $\clo M$ a union of $a$-many pairwise disjoint \zerosets is a \oneset.
		\item if $\clo M$ is compatible with some $\crel$, 
                  then for each $f$ in $\clo M$ a union of $c$-many pairwise disjoint \onesets is a \zeroset.
	\end{itemize} 
\end{lemm}
\begin{proof}
  The proofs of the two cases are analogous, so we will only prove the second
  one. Let $U_1, \ldots, U_c$ be disjoint \onesets of the $n$-ary function
  $f \in \clo M$ and $U = \bigcup_{i=1}^{c}U_i$.
  Since every coordinate $i$ occurs in exactly one set of
  $U_1, \ldots, U_{c},\compl U$ and $f$ is compatible with $\crel$,
  the tuple $(f(U_1), \ldots, f(U_c),f(\compl U))$ cannot
  evaluate to $(1,\ldots,1)$. Therefore $f(\compl U) = 0$ and $U$ is a \zeroset.
  See Figure~\ref{fig:conesets} for example.

  \begin{figure}[h]
    \begin{center}
      \tablexd
      \caption{Example of $c$ disjoint \onesets creating a \zeroset
        with $c=3$.
        The yellow column represents the result of an evaluation of function $f$ on tuples represented by other columns.
        The columns are in $\rel{1,{0,1,2,3},4}$ and 
        the grey cells are $U_1, \ldots, U_c$ while the red cells are $U$.
    }
      \label{fig:conesets}
    \end{center}
  \end{figure}
\end{proof}

\begin{proof}[Proof of Proposition~\ref{antichain-equivalent}]
	By using Lemma~\ref{onesettozeroset} we conclude that:
	\begin{itemize}
		\item if $f$ contains an antichain of \onesets of size $n$ then it also contains an antichain of \zerosets of size at least $\floor{\frac{n}{c}}$
		\item if $f$ contains an antichain of \zerosets of size $n$ then it also contains an antichain of \onesets of size at least $\floor{\frac{n}{a}}$
	\end{itemize}
	so if one of the antichains of \zerosets or \onesets is bounded then the other one has to be bounded as well.
\end{proof}

\begin{proof}[Proof sketch of Proposition~\ref{get-xor-at}] 
  Since $\clo M$ has unbounded antichains, 
  we can take a function from $\clo M$ with a arbitrarily large antichain of \onesets. 
  By taking its minor, we obtain $f$ satisfying
  \[
  f(1,0,\ldots,0) = f(0,1,0,\ldots,0) = \cdots = f(0,\ldots,0,1,0) = 1.
  \]
  Notice that the last coordinate is exceptional, it does not have to
  form a \oneset. 
  By taking further minors of $f$ we either get $g$,
  of arbitrarily large arity,
  that satisfies
  \[
    g(1,0,\ldots,0) = g(0,1,0,\ldots,0) = \cdots = g(0,\ldots,0,1) = 1,
  \]
  or 
  compatibility with $\AT$~%
  (see \fullversion). 
  We are left with the case when $g$'s, of arbitrarily large arity, are in $\clo M$.

  If $\clo M$ does not include \AT, it is
  compatible~%
  (after possibly changing ones to zeros and zeros to ones) with $\rel{1,{0,\dotsc,n-2,n},n}$ or
  $\rel{{0,d},{0,\dotsc,n-1},n}$ for some $d < n$. 
  We use these relational pairs for forcing further behavior of $g$, 
  and finally obtain an xor of an arbitrarily large arity. 
  This implies that \Xor is a subset of $\clo M$.
\end{proof}

\begin{proof}[Proof sketch of Proposition~\ref{get-small-fixing}]
  If a minion $\clo M$ has bounded antichains and does not have
  $q$-threshold for any $q$, we can find~%
  (skipping an easy case discussed in \fullversion%
  )
  positive integers $a,b,c,d$ such that $c/d < a/b < 1$ such that
  $\clo M$ is compatible with relational pairs
  \begin{equation}\label{eq}
    \rel{a,{0,\ldots,b-1},b},\quad\rel{c,{1,\ldots,d},d}.
  \end{equation}
  Notice that the converse, i.e. that these relational pairs prevent threshold, is
  clear since (\ref{eq}) disallows any $q$-threshold such that $q < a/b$ and
  any $q$-threshold such that $q > c/d$. 
  It can be shown that
  these relational pairs are the general obstacle to a threshold
  polymorphism. We prove the proposition by induction on $a+b+c+d$.

  For the reminder of the proof to work we are forced to work
  with weaker assumptions -- instead of $\clo M$ being
  compatible with (\ref{eq}) we assume that $\clo M$ is
  ``almost compatible'' with the relational pairs. Nevertheless, the
  ``almost compatibility'' notion
  is
  rather technical, and we ignore it in this sketch.
  For a formal proof, see
  \fullversion
  Here, let us simply assume that $\clo M$ is
  compatible with (\ref{eq}).

  It turns out that the only interesting case is $c/d < a/b < 1/2$. All
  the other cases can be either resolved directly or reduced to this
  one. Now, consider a minimal (ordered by inclusion) \zeroset $U$ and let
  $f_U$ denote $|U|$-ary operation obtained from $f$ by plugging zeros
  to every coordinate not contained in $U$.
  Since $f$ is compatible with $\rel{a,{0,\ldots,b-1},b}$ and $U$ is a
  \zeroset, $f_U$ is compatible with $\rel{c,{1,\ldots,d-c},d-c}$.
  Every \oneset in $f_U$ is also a \oneset in $f$, so $f_U$ has
  bounded antichains of \onesets.
  (bounded across every $f\in\clo M$ and every $U$).
  Moreover, since $U$ is minimal, the complement $\compl U$ of $U$ is
  ``almost'' a \oneset (every strict superset is).
  If $\compl U$ was a \oneset, $f_U$ would be compatible with
  $\rel{a,{0,1,\ldots,b-a-1},b-a}$ since $f$ is compatible with
  $\rel{a,{0,1,\ldots,b-1},b}$. This is where the weaker notion of
  compatibility (the star-compatibility) is necessary in the full proof. However for the sake of simplicity, assume that $f_U$ is
  compatible with $\rel{a,{0,1,\ldots,b-a-1},b-a}$. Since
  $f_U$ has bounded antichains of \onesets and it is compatible with
  relational pairs
  \[
  \rel{a,{0,1,\ldots,b-a-1},b-a}, \quad \rel{c,{1,\ldots,d-c},d-c}
  \]
  where $c/(d-c) < a/(b-a)$, it has also bounded antichains of \zerosets.
  Therefore, we can apply the induction hypothesis and
  obtain a small (bounded across every $f\in\clo M$ and every $U$)
  \onefset or \zerofset $V$ in $f_U$. For our purposes, we don't need
  to know that the set is fixing, it suffices that it is a \zeroset
  or a \oneset. Let $\mathcal L_f$ denote the set of all possible sets
  $V$ above across all minimal \zerosets $U$. From the induction
  hypothesis, we also get that either every $V\in\mathcal L_f$ is a
  \oneset in the appropriate $f_U$, or every $V\in\mathcal L_f$ is a
  \zeroset in the appropriate $f_U$.

  \begin{claim}
    The size of pairiwise disjoint subsystems of $\mathcal L_f$ is
    bounded by a number independent of the chosen $f\in\clo M$.
  \end{claim}
  If every $V\in\mathcal L_f$ is a \oneset in the appropriate $f_U$, then $V$ is a
  \oneset in $f$ and the claim follows from $\clo M$ having bounded
  antichains. Let us prove the claim if every $V\in\mathcal L_f$ is a \zeroset in
  the appropriate $f_U$. Consider $c$ disjoint elements
  $V_1, \ldots, V_c\in\mathcal L$, and let $U_1, \ldots, U_c$ be the
  appropriate minimal \zerosets. Thus also every $\compl {U_i}\cup V_i$
  is a \zeroset. Since
  \[
  U_1, U_2, \ldots, U_c, \compl {U_1}\cup V_1, \compl {U_2}\cup V_2,
  \ldots, \compl {U_c}\cup V_c
  \]
  are \zerosets, $V_1\cup \ldots\cup V_c$ is a \oneset by compatibility with $\rel{c,{1,2,\ldots 2c+1},2c+1}$. Let $M$ be the bound on
  antichains of \onesets in $\clo M$, the size of antichains in
  $\mathcal L$ is bounded by $cM$.

  Finally, we use the claim to find a small \onefset in $f$. Consider any maximal
  sequence $V_1,\ldots,V_n\in\mathcal L$ of disjoint sets and let
  \[
  W = V_1\cup V_2\cup \ldots\cup V_n,
  \]
  Every \zeroset contains a minimal \zeroset, every minimal \zeroset
  contains some $V\in\mathcal L$ and every $V\in\mathcal L$ intersects $W$.
  Therefore every \zeroset intersects $W$, so $W$ is the desired
  \onefset.
\end{proof}

%% file: ManyOneSets.tex
\section{Proof of Proposition~\ref{get-xor-at}}\label{app:get-xor-at}
	Let \maintemplateM be a symmetric template, such that $\clo M = \pol{\relstr A,\relstr B}$ is idempotent.\\
	If $\clo M$ is compatible with some $\arel$, some $\crel$ \\
	and does not have bounded antichains then $\clo M$ includes $\Xor$ or $\AT$.\\

\begin{defn}
	Let \maintemplateM be a Boolean template. We define $flip \maintemplate = \neg \maintemplate$ (each zero converted to one and vice versa).
\end{defn}
\begin{defn}
	Let $\clo M$ be a minion. We define \[\flip{\clo M} = \set{f : \exists g \in \clo M \ \forall x \in \set{0, 1}^{arity(g)} \  f(x) = \neg g(\compl x)}\]
\end{defn}
\begin{corollary}
	Let \maintemplateM be a Boolean template, and $\clo M = \pol \maintemplateflatten$. Then $\flip{\clo M} = \pol{\flip{\maintemplateflatten}}$
\end{corollary}
\begin{defn}
	Let $f : \set{0, 1}^L \rightarrow \set{0, 1}$. Function $f$ is folded, iff \[ \forall x \in \set{0, 1}^{L} \ f(x) = \neg f(\compl x) \]
\end{defn}
We provide following notations: 
\begin{itemize}
	\item For function $f$, $O_f(X)$ is image of set $X$ under $f$.
	\item For family of functions $F$, $O_F(X) = \bigcup_{f \in F} O_{f}(X)$. 
	\item For $I \subseteq \fullset{k}$, $I_k = \set{t: \text{t is k-ary tuple of Hamming weight }h,\ h \in I}$.
\end{itemize}

\begin{claim}\label{c:flipisok}
	Proposition~\ref{get-xor-at} holds for symmetric template \maintemplateM iff it holds for $ \flip{\maintemplateflatten} $.
\end{claim}
\begin{proof}$ $
	\begin{itemize}
		\item If \maintemplateM satisfies assumptions of our proposition, then $\flip{\maintemplateflatten}$ also does: \\
		$\flip{\clo M} = \pol{\flip{\maintemplateflatten}} $ is compatible with ${\rel{c,{1,\dotsc,c+1},c+1}}$ and ${\rel{1,{0,\dotsc,a},a+1}}$, 
		so compatibility with $some$ $\arel$ and $\crel$ holds ($a$ and $c$ swapped). Moreover, $\flip{\clo M}$ does not have bounded antichains by Proposition~\ref{antichain-equivalent}.
		\item Conclusion for $\flip{\maintemplateflatten}$ implies conclusion for \maintemplateM: both $\AT$ and $\Xor$ are folded, so they are preserved under $flip$.
	\end{itemize}

\end{proof}
From this Claim, later in the proof, we can flip $\Gamma$, and therefore $\clo M$ without loss of generality.

Assume that $\clo M$ does not have \AT. Our goal is to prove that with such assumption $\clo M$ has \Xor.
Recall Claim 4.5 from ~\cite{PCSPSoda}:
\begin{claim}\label{c:atorbit}
	Consider $k \ge 1$ then
	\begin{enumerate}
		\item $\orbits{\AT}{0}{k} = \srel{0}{k}$,
		\item $\orbits{\AT}{k}{k} = \srel{k}{k}$,
		\item $\orbits{\AT}{0, k}{k} = \srel{0, k}{k}$,
		\item $\orbits{\AT}{l}{k} = \srel{1, \ldots, k-1}{k}, k \ge 2, l \in \set{1, \ldots, k-1}$,
		\item $\orbits{\AT}{l_1, l_2}{k} = \fullseti{k}{k}$ if $k \ge 2$ and $\set{l_1, l_2} \neq \{0, k\}$.
	\end{enumerate}
\end{claim}

\begin{lemma}[If no AT]\label{l:if-no-at}
	Let \maintemplateM symmetric idempotent template incompatible with \AT. $\clo M = \pol{\maintemplateflatten}$. Then $\clo M$ or $\flip{\clo M}$ is compatible with
	\noextermes or \middlegap.
\end{lemma}
Note: Lemma 4.4 of~\cite{PCSPSoda} is a version of this lemma for folded case (template compatible with folded minion in assumption, "$\clo M$" instead of "$\clo M$ or $\flip{ \clo M}$" in conclusion). In fact, proof of this lemma is identical with accuracy to one step: in case of Lemma 4.4 of~\cite{PCSPSoda} they can flip \maintemplateM and still be compatible with $\clo{M}$, because $\clo{M}$ was folded. In our case if we flip \maintemplateM we flip $\clo{M}$, so we need "$\clo M$ or $\flip \mainminion$" in our conclusion. We provide very brief proof of this lemma, for more details we refer reader to Lemma 4.4 of~\cite{PCSPSoda}.
\begin{proof}
	Take $\rel{I, J, k}$ from language of \maintemplateM incompatible with \AT. So $O_{AT}(I) \not\subseteq J $. By Claim \ref{c:atorbit}, one of the following cases holds:
	\begin{enumerate}
		\item  $ I = \srel{l}{k}, k \ge 2, l \in \set{1, \ldots, k-1} \land \srel{1, \ldots, k-1}{k} \not\subseteq J$ \\
		in which we conclude that  $\clo M$ or $\flip \mainminion$ is compatible with \middlegap.
		\item  $\srel{l_1, l_2}{k} \subseteq I, k \ge 2, \{l_1, l_2\} \neq \{0, k\} \land \fullset{k} \not\subseteq J$
		in which we conclude that  $\clo M$ or $\flip \mainminion$ is compatible with \noextermes.
	\end{enumerate}
	We prove only first case, as reasoning in second is almost identical.
	There is ${b \in \set{1, \ldots, k-1} \setminus \set l}$ such that $b \not\in J$. Now, there are 2 cases: $b > l$ or $b < l$. So modulo flip $b > l$.
	Now using standard relaxations we obtain that $\clo M$ or $\flip \mainminion$ is compatible with \mbox{\middlegap.}
	
\end{proof}
From Lemma~\ref{l:if-no-at} and Claim~\ref{c:flipisok} we can assume that $\clo M$ is compatible with \middlegap or \noextermes.
Choose maximal $F\subseteq \clo M$ such that for every $f\in F$ we have \[f(1,0,\dotsc,0)=f(0,1,0,\dotsc,0)=\dotsb=f(0,\dotsc,0,1,0)=1\]
Such functions of arbitrarily large arity can be obtained by taking a function with sufficiently large antichain of \onesets and identifying variables inside each \oneset from this antichain.
For every function with such property we can obtain a function with arbitrary smaller arity and same property by identifying last positions. Therefore $F$ contains functions of all arities.

Choose maximal $G\subseteq F$ such that for every $g\in G$ we have \[g(1,0,\dotsc,0)=g(0,1,0,\dotsc,0)=\dotsb=g(0,\dotsc,0,1)=1\]

\noindent Now reasoning split into 2 cases, dependent on arities of functions from $G$
\begin{description}
	\item{CASE 1:}\label{c:Gbound} Arities of functions from $G$ are bounded by a constant $M$ - results in obtaining \AT which leads to contradiction
	\item{CASE 2:}\label{c:Gunbound} $G$ contains functions of arbitrarily large arity - results in obtaining \Xor
\end{description}

\subsection{CASE: Arities of functions from $G$ are bounded by a constant $M$}\label{a:boundedarities}

Take any $f\in F$ with sufficiently large arity,
partition variables into sets of size $M$ with the last partition possibly smaller.
We identify variables in each set to obtain $f'(x_1,\dotsc,x_n,x_{n+1})$
and note that if $U\neq \{1,\dotsc,n+1\}$ and $(n+1)\in U$ then $f'(U) = 0$~%
(otherwise modifying $f$ by identifying variables in $U$,
after unfolding the partition,
we get a contradiction with the choice of $M$).

Our goal is to obtain almost negations of all arities.
An {\em almost negation} is function of arity $m$ defined as:
\[
f(U) =
\begin{cases}
1&\text{ if } U = \{1,\dotsc,m\} \\
0&\text{ if } U = \emptyset \\
1-x_m&\text{ else.}
\end{cases}
\]
We provide such notation: $\AN$ is family of almost negations of all arities, $\AN_m$ is almost negation of arity m.

Note that the reasoning in previous paragraph provides~(in $\clo M$) functions of all arities
satisfying the condition of being $\AN_m$ whenever $x_m=1$. Consider the $f'$ from few paragraphs above. Function $f'$ is compatible with $\arel$ and therefore a union of $a$-many pairwise disjoint \zerosets is a \oneset (by Lemma \ref{onesettozeroset}).
Then if $(n+1)\notin U$ and $|U|\geq a$ then $f'(U) = 1$, because any nonempty subset of $\{1, \dots, n\}$ is a \zeroset (by a property established few paragraphs above) so we can split $U$ into $a$ non-empty disjoint \zerosets.
That means that identifying variables into sets of size $a$ (to get rid of $|U| \ge a$ constraint), with the last set of possibly different size,
we obtain an almost negation.
As $F$ had functions of arbitrary arity, we obtain almost negations of arbitrarily large arities and~%
(by identifying variables) of arbitrary arities in the end.
This case is concluded by the following lemma.
\begin{lemma}\label{l:anisat}
	Let $\relstr A,\relstr B$ be a symmetric template.
	$\relstr A,\relstr B$ is compatible with $\AN$ iff it is compatible with $\AT$.
\end{lemma}
\begin{proof}
	$\impliedby$:\\
	Let $at$ denote the alternating threshold of arity $2k +1$.
	\[\AN_{k+2}(x_1, \ldots, x_{k+2}) = at(\underbrace{x_{k+2}, \ldots, x_{k+2}}_{\text{k times}}, x_1 \ldots, x_{k+1})\]
	So $\AN_{k+2}$ is minor of $at$, and since each $\AN_{k}$ is minor of $\AN_{k+1}$, then our conclusion holds.
	$\implies$:\\
	All we need to prove is that $O_{\AT}(I) \subseteq O_{\AN}(I)$ for each symmetric relation I.
	Fix I. Let $k \ge 1$ be arity of this relation.
	Then from Claim \ref{c:atorbit}:
	\begin{numcases}{O_{\AT}(I)=}
	\set{0}_k & $I = \set{0}_k$ \\
	\set{k}_k & $I = \set{k}_k$ \\
	\{0, k\}_k & $I = \{0, k\}_k$\\
	\{1, \ldots k-1\}_k &$ I = \{l\}, k \ge 2,  l \in \{1, \ldots k-1\}$ \label{c:oneell} \\
	\fullseti{k}{k} & $\{l_1, l_2\}_k \subseteq I, k \ge 2, \{l_1, l_2\} \neq \{0, k\} \label{c:twoel}$
	
	\end{numcases}
	Cases 1 - 3 trivially holds for any idempotent function. Now:
	\begin{itemize}
		\item if our thesis holds for (\ref{c:oneell}), then in (\ref{c:twoel}) it is sufficient to show that $ \srel{0, k}{k} \subseteq O_{\AN}(I)$
		\item we can reduce (\ref{c:twoel}) to case, when $I = \set{l_1, l_2}$
	\end{itemize}
	Thus it is sufficient to show that:
	\begin{enumerate}
		\item $\srel{1, \ldots, k-1}{k} \subseteq O_{\AN}(I)$ for $ I = \srel{b}{k}, k \ge 2, b \in \set{1, \ldots, k-1} $
		\item $\srel{0, k}{k} \subseteq O_{\AN}(I)$  for $I = \{l_1, l_2\}_k, k \ge 2, \{l_1, l_2\} \neq \{0, k\}$
	\end{enumerate}
	
	\noindent CASE 1:
	Choose function $f \in$ \AN of arity $n = 2k +1 $.
	Fix any $r$ such that $b \le r < k$. Let $d = r - b + 1$, $U = \compl{\set{1, \ldots, d}}$, then: 
	\[\bigl(\applyshort{f}{1, \ldots, n}{b-1},\ f(\set 1),\ \ldots\ ,\ f(\set{d}),\ f(U),\ \applyshort{f}{\emptyset}{k-r-1}\bigr)\]
	is tuple of Hamming $r$. Producing tuples of Hamming weight less than $b$ and at least $1$ is symmetric, since almost negation is folded.
	Example of this reasoning shows Figure \ref{f:anis}.
	\centfig{\anisat}{Example of producing $ 6_9 $-tuple from $ 3_9 $-tuples by $\AN_9$}{f:anis}

	\noindent CASE 2: \\
	Assume $l_1 < l_2$.
	Let: $d = k - l_1$,  $w = l_2 - l_1$, 	$t_1,  \ldots t_d$ be $k-ary$ tuples of Hamming weight $l_2$, such that (numbering of positions from zero):
	\[
	t_1(i) =
	\begin{cases}
	0&\text{ if } i \in [0, d - w - 1]\\
	1&\text{ else.}
	\end{cases}
	\]
	And for $1 < j \le d$
	\[
	t_j(i) =
	\begin{cases}
	t_{j-1}[i - 1 \mod d]&\text{ if } i < d\\
	1&\text{ else.}
	\end{cases}
	\]
	Let $r$ be $k-ary$ tuples of Hamming weight $l_1$, such that:
	\[
	r(i) =
	\begin{cases}
	0&\text{ if } i \in [0, d-1]\\
	1&\text{ else.}
	\end{cases}
	\]
	Choose function $f \in$ \AN of arity $d+1$.
	Then $f(t_1, \ldots, t_d, r)$ = $(1, \ldots, 1)$. 
	From the fact that almost negation is folded, obtaining $(0, \ldots, 0)$ is symmetric.
	Example of this reasoning shows Figure \ref{f:aniss}.
	\centfig{\anisatxd}{Example of producing $ 6_6 $-tuple from $\set{4_6, 2_6 }$-tuples by almost negation of arity 5}{f:aniss}

\end{proof}

\subsection{CASE: Arities of functions from G are unbounded}\label{a:unboundedarities}
First goal is to show that there is some constant $e$ such that every g in $G$ is \flippable{e}.

\begin{defn}
	A $k$-ary function $g$ is \oneflippable{e} iff for every $U$, such that $g(U) = 1$ and $k-e > |U| > e$ and for every $x\not\in U$, $y \in U$:
	\[g(U \setminus \set{y}) = g(U\cup\{x\}) = 0\]
\end{defn}
\begin{defn}
	A $k$-ary function $g$ is \zeroflippable{e} iff for every $U$, such that $g(U) = 0$ and $k-e > |U| > e$ and for every $x\not\in U$, $y \in U$:
	\[g(U \setminus \set{y}) = g(U\cup\{x\}) = 1\]
\end{defn}
\begin{defn}
	A function $g$ is \flippable{e} iff it is both \zeroflippable{e} and \oneflippable{e}.
\end{defn}
\begin{corollary}
	$k$-ary function $g$ is \flippable{e} iff for every $U$, such that $k-e > |U| > e$ and coordinates $x\not\in U$, $y \in U$:
	\[g(U \setminus \set{y}) \not= g(U) \not=g(U\cup\{x\})\]
\end{corollary}

Recall that $ G $ is compatible with \middlegap or \noextermes. Firstly we show that each $g \in G$ is \oneflippable{n}. To do that we will prove following two lemmas:

\begin{lemm}\label{l:oneflippablefirstcase}
	Let g be k-ary idempotent function, where $k > 2n$ such that:
	\[g(1,0,\dotsc,0)=g(0,1,0,\dotsc,0)=\dotsb=g(0,\dotsc,0,0,1)=1\]
	if $g$ is compatible with \middlegap then $g$ is \oneflippable{(n-1)}.
\end{lemm}
\begin{proof}

	Fix \oneset $U$  of g such that $ k - (n-1) > |U| > n-1$. Reasoning splits into 2 cases:
	\begin{itemize}
		\item CASE 1: $f(U \setminus \set{i}) \stackrel{?}{=} 0$ for all $i \in U$
		\item CASE 2: $f(U \cup \set{j}) \stackrel{?}{=} 0$ for all $j \not\in U$
	\end{itemize}
	\begin{proof}[Proof of CASE 1 ($U \setminus \set{i}$)]\hfill\\
		Without loss of generality assume that $i = n-2$ and  it is the first element of set $U$. Let $R = \compl{U \cup \set{1,\dotsc,n-3}}$.\\
		$g(R) = 0$ because otherwise tuple:
		\[\bigl(g(\set{1}),\ g(\set{2}),\ \dotsc,\ g(\set{n-3}),\ g(\emptyset),\ g(U),\ g(R)\bigl)\]
		 have Hamming weight $n-1$ which contradicts compatibility with \middlegap.
		 Removing $n-2$ from from $U$ and adding it to empty set and applying $g$ will produce tuple:
		 \[\bigl(g(\set{1}),\ g(\set{2}),\ \dotsc,\ g(\set{n-3}),\ g(\set{n-2}),\ g(U \setminus \set{n-2}),\ g(R)\bigr)\]
		\noindent which implies that $g(U \setminus \set{n-2}) = 1$, because otherwise this tuple would have Hamming weight $n-1$ which contradicts compatibility with \middlegap.
		Example of this reasoning is shown in Figure \ref{fig:flipping0}.
			\begin{figure}[h]
			\begin{center}
				\flippingonetable
			\end{center}
			\caption{Scheme for forcing $f(U \setminus \set{i})=0$ when $f(U) = 1 \land i \in U$.\\
			Red cells in first table are $ U $, in second are $ U \setminus \set{i} $, brown cell is $ \set i $, gray cells in last row are $ R $, gray cells on diagonal are $\set{1}, \ldots, \set{n-3}$, yellow cells represent result of applying $ g $.
			}
			\label{fig:flipping0}
		\end{figure}
	\end{proof}
	\begin{proof}[Proof of CASE 2 ($U \cup \set{j}$)]\hfill\\
		Without loss of generality assume that $j = n-2$ and all elements of $U$ are greater than $n-2$. let $R = \compl{U \cup \set{1,\dotsc,n-3}}$.\\
		$g(R) = 1$ because otherwise tuple:
		\[\bigl(g(\set{1}),\ g(\set{2}),\ \dotsc,\ g(\set{n-3}),\ g(\set{n-2}),\ g(U),\ g(R)\bigr)\]
		have Hamming weight $n-1$ which contradicts compatibility with \middlegap.
		Removing $n-2$ from it's singleton and adding it to $U$ and applying $g$ will produce  tuple:
		\[\bigl(g(\set{1}),\ g(\set{2}),\ \dotsc,\ g(\set{n-3}),\ g(\emptyset),\ g(U \cup \set{n-2}),\ g(R)\bigr)\]
		\noindent which implies that $g(U \cup \set{n-2}) = 0$, because otherwise this tuple would have Hamming weight $n-1$ which contradicts compatibility with \middlegap.
		Example of this reasoning is shown in Figure \ref{fig:flipping1}.
			\begin{figure}[h]
			\begin{center}
				\flippingzerotable
			\end{center}
			\caption{Scheme for forcing $f(U \cup \set{j})=0$ when $f(U) = 1 \land j \notin U$.\\
				Red cells in first table are $ U $, brown cell is $ \set j $, gray cells in last row are $ R $, gray cells on diagonal are $\set{1}, \ldots, \set{n-2}$, yellow cells represent result of applying $ g $.
			}
			\label{fig:flipping1}
		\end{figure}
	\end{proof}
	So since we prove that in both cases $g(U \cup \set{j}) = g(U \setminus \set{i})=0$ then g is \oneflippable{(n-1)}.
	\end{proof}

	\begin{lemm}\label{l:oneflippablesecondcase}
		Let g be k-ary idempotent function, where $k > 2n$ such that:
		\[g(1,0,\dotsc,0)=g(0,1,0,\dotsc,0)=\dotsb=g(0,\dotsc,0,0,1)=1\]
		if $g$ is compatible with \noextermes then $g$ is \oneflippable{$\frac{n}{d}$}.
	\end{lemm}
	\begin{proof}
	Notice that $d\nmid n$, because otherwise we can produce tuple
	\[\bigl(\applyshort{g}{1}{d},\ \applyshort{g}{2}{d},\ \dotsb,\ \applyshort{g}{\frac{n}{d}}{d}\bigr)\]
	containing only ones, which contradict compatibility with \noextermes.\\
	Let $s = \floor{n/d}, r = n \ \text{mod } d$. Now our proof will split into same cases as in
	Lemma~\ref{l:oneflippablefirstcase}:
	\begin{itemize}
		\item CASE 1: $f(U \setminus \set{i}) \stackrel{?}{=} 0$ for all $i \in U$
		\item CASE 2: $f(U \cup \set{j}) \stackrel{?}{=} 0$ for all $j \not\in U$
	\end{itemize}
	\begin{proof}[Proof of CASE 1 ($U \setminus \set{i}$)]\hfill\\
		Without loss of generality, assume that $i = s$ and all elements of $U$ are not smaller than $s$. Now we can produce tuple:
		\[\bigl(\applyshort{g}{1}{d},\ \dotsb,\ \applyshort{g}{s-1}{d},\ \applyshort{g}{s}{r},\ \applyset{g}{U}{d-r},\ \applyset{g}{U \setminus \set{s}}{r}\bigr)\]
		Notice that all elements of this tuple, except $g(U \setminus \set{s})$ evaluate to 1, so $g(U \setminus \set{s}) = 0$, because otherwise we break compatibility with \noextermes.
		See Figure~\ref{fig:noextremesflipzero} for example.
		\begin{figure}[h]
			\begin{center}
				\noextremesflipone
			\end{center}
			\caption{Scheme for forcing $f(U \setminus \set{i})=0$ when $f(U) = 1 \land i \in U$.\\
			Red cells in first rows are $U$, red cells in last row are $ U \setminus \set i $, brown cell represents lack of $ \set i $, gray cells in first rows are 
			$ \set 1 , \ldots, \set{s-1} $, lonely gray cell is $ \set s $, yellow cells are result of applying $ g $.
			}
			\label{fig:noextremesflipzero}
		\end{figure}
	\end{proof}

	\begin{proof}[Proof of CASE 2 ($U \cup \set{j}$)]\hfill\\
		Without loss of generality, assume that $j = s$ and all elements of $U$ are larger  than $s$. Now we can produce tuple:
		\[\bigl(\applyshort{g}{1}{d},\ \dotsb,\ \applyshort{g}{s-1}{d},\ \applyshort{g}{s}{r},\ \applyset{g}{U \cup \set{s}}{d-r},\ \applyset{g}{U}{r}\bigr)\]
		Notice that all elements of this tuple, except $g(U \cup \set{s})$ evaluate to 1, so $g(U \cup \set{s}) = 0$, because otherwise we break compatibility with \noextermes.
		See Figure~\ref{fig:noextremesflipone} for example.
		\begin{figure}[h]
			\begin{center}
				\noextremesflipzero
			\end{center}
			\caption{Scheme for forcing $f(U \cup \set{j})=0$ when $f(U) = 1 \land j \notin U$.\\
				Red cells are $ U $, brown cell is $ \set j $, gray cells in first rows are 
				$ \set 1 , \ldots, \set{s-1} $, lonely gray cell is $ \set s $, yellow cells are result of applying $ g $.
			}
			\label{fig:noextremesflipone}
		\end{figure}
	\end{proof}

\end{proof}

From lemmas \ref{l:oneflippablefirstcase} and \ref{l:oneflippablesecondcase} we obtain that each $g \in G$ is \oneflippable{n}. Applying this result to following lemma gives us that each $g \in G$ is \flippable{e}, where $e = n + c(a-1)$.  

\begin{lemm}\label{l:zeroflippable}
	Let g: k-ary \oneflippable{n} function, such that every $U$ of size c is a \zeroset of g. If g is compatible with $\arel$ then g is \flippable{e}, where $e = n + c(a-1)$.
\end{lemm}
\begin{proof} Our aim is to show that $g$ is \zeroflippable{e}, then our conclusion follows immediately.
	Let $U$ be a \zeroset such that $\flippableassumption{U}{k}$, $V$ be a set such that $U, V$ differs on one element $i$. 
	If we show that $\compl V$ is a \oneset, then g is \oneflippable{e}.\\ 
	Pick $S_1, ..., S_{a-1} \subseteq{\compl{U \cup V}}$ - disjoint sets of size c. It is possible, because $|\compl{U \cup V}| \ge c(a-1)$. Let $W = S_1 \cup ... \cup S_{a-1} \cup U$. Note that $S_1, ..., S_{a-1}$ are \zerosets because they have size $c$.
	From compatibility with $\arel$ union of $a$ disjoint \zerosets is a \oneset (by Lemma \ref{onesettozeroset}), so $W$ is \oneset. \\
	Let $W' = S_1 \cup ... \cup S_{a-1} \cup V$. Now:
	\begin{itemize}
		\item $W$ is a \oneset
		\item $|W| = |U| + c \cdot (a-1)) > n$
		\item $|\compl{W}| = |\compl{U}| - c(a-1) > n$
		\item $W'$ and $W$ differs only on $i$.
		\item $g$ is \oneflippable{n}
	\end{itemize}
	So $\compl{W'}$ is \zeroset.
	$\compl{V} = S_1 \cup ... \cup S_{a-1} \cup \compl{W'}$, thus $\compl{V}$ is an union of $a$ disjoint \zerosets, so it is a \oneset.\\
	See Figure~\ref{fig:zeroflipping} for example where $|V| > |U|$.

	\begin{figure}[h]
		\begin{center}
			\tablexda
		\end{center}
		\caption{Scheme for forcing $f(\compl V)=1$\\
			Red cells are $U$, brown cell is $i$, gray zeros are \iterm{S}{1}{a-1}, gray ones are $W$,  yellow cells are result of applying $g$.}
		\label{fig:zeroflipping}
	\end{figure}

\end{proof}
To finish proof of CASE 2 of Proposition~\ref{get-xor-at} we prove following two lemmas:

\newcommand{\callplussize}[2]{{#1}(#2) + |#2|}
\begin{lemm}\label{l:flip}
	
	Let g: \flippable{e} n-ary function such that $n \ge 3(e + 1)$
	Let U, V: sets such that \flippableassumptionM{U}{n} and \flippableassumptionM{V}{n}\\
	Then $\callonset{g}{U} = \callonset{g}{V} \iff |U| \equiv_2 |V|$.
\end{lemm}
\begin{proof}
	Conclusion in our Lemma is equivalent to: $\callplussize{g}{U} \equiv_2 \callplussize{g}{V}$.
	Case, when $V \subseteq U$ is easy - induction on Hamming distance between $U$ and $V$.\\
	\noindent CASE $|U \cup V| < n - e$:\\
	Because $|U \cup V| > e$ and our thesis holds if one set is subset of another, then:
	\[
	\callplussize{g}{U} \equiv_2
	\callplussize{g}{U \cup V} \equiv_2
	\callplussize{g}{V}
	\]
	\noindent CASE $|U \cup V| \ge n - e \ge 2(e+1)$:\\
	We can pick subsets: $U' \subseteq U, V' \subseteq V$, such that $|U'| = |V'| = e + 1$. Now:
	\[
	\callplussize{g}{U} \equiv_2
	\callplussize{g}{U'} \equiv_2
	\callplussize{g}{U' \cup V'} \equiv_2
	\callplussize{g}{V'} \equiv_2
	\callplussize{g}{V}
	\]
	
\end{proof}

\begin{lemm}\label{l:flippabletocontr}
	Let \maintemplateM - symmetric idempotent template, e - some natural number, G~-~infinite set of \flippable{e} functions.\\
	If \maintemplateM is compatible with G, then it is compatible with: $\Min, \Max, \Xor, \text{or } \AT$.
\end{lemm}
\begin{proof}
	
	In this proof, if for function $f$, and number $k$ following holds: \\$\forall_{U, V: |U| = |V| = k} \ \callonset{f}{U} = \callonset{f}{V}$,
	then we use such abbreviation: $f(k) = \callonset{f}{V}$ for some $V$ such that $|V| = k$.
	
	Since $g$ is \flippable{e} then from Lemma \ref{l:flip}: \[\forall_{x, x'} n-e > x > e \land n-e > x' > e \implies g(x) \equiv_2 g(x') + x + x' \]
	Let $G_{even} = \set{g: g \in G \land arity(g) = 2k}$, $G_{odd} = \set{g: g \in G \land arity(g) = 2k + 1}$.
	At least one of $G_{odd}$, $G_{even}$ must be infinite.\\
	\textbf{CASE \boldmath{$G_{even}$} is infinite:}\\
	Pick $g$ from $G_{even}$ with sufficiently large arity. Obtain $g'$ by partitioning variables of $g$ into sets of even size between $e +1$ and $2e +2$. Now, by Lemma~\ref{l:flippabletocontr} there is some $c \in \set{0,1}$ such that $g'(x) = c$ if $x \not\in \set{0, arity(g')}$, and thus $g'$ is min or max ($g'$ cannot be constant from idempotency of \maintemplateM and thus $G$). Since we can pick g of arbitrary large arity,
	and $arity(g') \ge arity(g)/(2e +2) $, we have min or max of arbitrary large arity. So \maintemplateM is compatible with infinite family of mins or infinite family of maxes and thus, by taking minors, \maintemplateM is compatible with \Min or \Max.\\
	\textbf{CASE \boldmath{$G_{odd}$} is infinite:}\\
	Pick $g$ from $G_{odd}$ with sufficiently large arity n. Obtain $g'$ by partitioning variables of $g$ into odd number of sets of odd size between $e +1$ and $2e +2$.
	Let $k$ minimal natural number, such that $2k+1 \ge e +1$, $m =arity(g')$.
	Then (again by Lemma~\ref{l:flip}):
	\begin{itemize}
		\item if $g(2k+1) = 0$, then $g'$ is xor
		\item if $g(2k+1) = 1$, then
		\[
		g'(x) =
		\begin{cases}
		1&\text{ if } x = m \\
		0&\text{ if } x = 0 \\
		1-xor(x)&\text{ else.}
		\end{cases}
		\]
		Take minor of $g'$: $f'(x_1, x_2, \ldots, x_{(m-1)/2}, y) = g'(x_1,x_1,x_2,x_2,\dotsc,x_{(m-1)/2},x_{(m-1)/2},y)$.
		Function $f'$ is $\frac{m+1}{2}$-ary almost negation.
	\end{itemize}
	Since we can pick g of arbitrary large arity and $m \ge arity(g)/(2e +2) $, we have xor or almost negation of arbitrary large arity. So \maintemplateM is compatible with infinite family of odd-arity xors or infinite family of almost negations and thus, by taking minors, \maintemplateM is compatible with \Xor, or (by Lemma~\ref{l:anisat}) with \AT.
	
\end{proof}
Fact that each $g \in G $ is \flippable{e} together with Lemma~\ref{l:flippabletocontr} combined with assumptions made by Proposition~\ref{get-xor-at} and our proof, are sufficient to finish proof of Proposition~\ref{get-xor-at}, because we forbidden $\clo M$ from having \Min, \Max and \AT so it must has \Xor.

%% file: FewOneSets.tex
\section{Proof of Proposition~\ref{get-small-fixing}}\label{app:get-small-fixing}

\begin{prop}
	Let $\Gamma$ be a symmetric language such that $\clo M = \pol\Gamma$ is idempotent.
	If $\clo M$ has bounded antichains
	and does not include  any of $\THR q$ then it has
	small fixing sets.
\end{prop}

Let assume, that antichains in $\clo M$ are bounded by constant $n-2$.

\begin{prop}\label{p:fewdisjointsets}
	An operation $f$ is compatible with $\rel{1,{0,\dotsc,n-2},n}$ if and only if
	it has no antichain of size $n-1$.
\end{prop}
\begin{proof}
	\begin{description}
		\item{CASE $\implies$:} if there would be at least $n-1$ disjoint \onesets then in first $n-1$ rows we can put ones on these \onesets (for each row we would pick different \oneset) and all remaining ones in last row. Since ones in first $n-1$ rows are forming \onesets then applying $f$ would result in a tuple with at least $n-1$ ones -- contradiction.
		\item{CASE $\impliedby$:} if $f$ is not compatible with $\rel{1,{0,\dotsc,n-2},n}$ then there is a matrix witness of non-compatibility which has at least $n-1$ rows with ones forming a \onesets. Since each tuple contains only one one then all these \onesets are disjoint -- contradiction.
	\end{description}
\end{proof}

It easily follows, that $\clo M$ is compatible with $\rel{1,{0,\dotsc,n-2},n}$. We proceed to establishing a number of easy facts about thresholds of symmetric relations.

\begin{fact}
	The following hold:
	\begin{itemize}
		\item A relation $\rel{a,J,n}$ is compatible with all $q$-thresholds
		if and only if $\left(n-\frac{n-a}{1-q},\frac{a}{q}\right)\cap \{0,\dotsc,n\}$
		is a subset of $J$.
		\item If a relation $\rel{I,J,n}$ with $a\in I$ and $b\notin J$
		is compatible with all $q$-thresholds then either $b>a$ and $q\geq a/b$
		or $b<a$ and $q\leq\frac{a-b}{n-b}$.
	\end{itemize}
\end{fact}

For the rest of the proof we will assume, that there is no threshold in $\clo M$ with aim of proving that in such case $\clo M$ will have a small fixing set.

The only two possible reasons for not having any threshold are relations:
\begin{itemize}
	\item $\rel{I,J,n}$ and $\rel{K,L,m}$ in $\Gamma$ with $a\in I, b\notin J,c\in K,d\notin L$ such that $b>a$, $d<c$ and $(c-d)/(m-d)<a/b$.
	\item $\rel{I, J, n} \in \Gamma$ with $a < b < c$, such that $\set{a, c} \not= \set{0, n}$ and $a, c \in I, b \not\in J$
\end{itemize}

Using standard relaxations for the first case (removing numbers from left set, adding numbers to right set and moving to the right in case $\rel{I,J,n}$ and to the left in case $\rel{K,L,m}$) we conclude that $\clo M$ is compatible with $\rel{a,{0,\dotsc,b-1},b}$ and $\rel{c',{1,\dotsc,d'},d'}$ (where $c' =c-d$ and $d' = m-d$) with $c'/d' < a/b'$. In second case by using same relaxations and possible flipping the relation we conclude that $\clo M$ is compatible with $\rel{{1, l}, \set{0,\dotsc, k-1, k+1, \dotsc l},l}$.
This reduces a problem to proving the following propositions:

\begin{prop}\label{p:gap}
	If $\clo M$ is compatible with $\rel{{1, l}, \set{0,\dotsc, k-1, k+1, \dotsc l},l}$ for some $k,l$ and have bounded antichains then it has small fixing set.
\end{prop}
\begin{proof}
	See Section \ref{a:gap}.
\end{proof}

\begin{prop}\label{p:fake}
	Let $a,b,c,d$ and $n$ be natural numbers such that $c/d<a/b$.
	The polymorphisms of the three relations
	\begin{itemize}
		\item $\rel{a,{0,\dotsc,b-1},b}$
		\item $\rel{c,{1,\dotsc,d},d}$ and
		\item $\rel{1,{0,\dotsc,n-2},n}$
	\end{itemize}
	have small fixing set.
\end{prop}
\noindent In order to proceed we need one more definition:
\begin{defn}
	A function $f$ is {\em compatible with $\rels{a,{0,\dotsc,b-1},b}$},
	if for every evaluation of $f$ on tuple of weight $a$ producing $b$ ones,
	in every row is a minimal \oneset~%
	(i.e. flipping any non-empty set of argument from $1$ to $0$ we change the result to $0$).
	
	Similarly $f$ is {\em compatible with $\rels{c,{1,\dotsc,d},d}$}
	if for every evaluation of $f$ on tuple of weight $c$ producing only zeros
	every row is minimal \zeroset.
\end{defn}
\begin{figure}[h]
	\begin{center}
		\starexample
	\end{center}
	\caption{Example of $f$ compatible with $\rels{2,{0,\dotsc,3},4}$. Because the result of applying $f$ is equal to $\mathbf{1}$, then all red sets must form minimal \onesets of $f$.}
	\label{fig:starexample}
\end{figure}
\noindent We will prove the following proposition, which obviously implies Propostion~\ref{p:fake}.
\begin{prop}\label{p:real}
	Let $a,b,c,d$ and $n$ be natural numbers such that $c/d<a/b$.
	The set of function compatible with
	\begin{itemize}
		\item $\rels{a,{0,\dotsc,b-1},b}$
		\item $\rels{c,{1,\dotsc,d},d}$ and
		\item $\rel{1,{0,\dotsc,n-2},n}$
	\end{itemize}
	has a small fixing set.
\end{prop}
\begin{proof}
	See Section~\ref{a:real}.
\end{proof}
\subsection{Proof of Proposition~\ref{p:gap}}\label{a:gap}
Assume that antichains are bounded by constant $M$ our aim is to prove, that $M$ have a fixing set of size at most $kM$.

\begin{lemm}
	If $f$ is compatible with $\rel{{1, l}, \set{0,\dotsc, k-1, k+1, \dotsc l},l}$, then $f$ does not have a minimal \zeroset of size greater or equal to $k$.
\end{lemm}
\begin{proof}
	Assume otherwise. Let $U$ be a \zeroset of $f$ such that $|U| \geq k$.  Without loss of generality assume that $U = \set{1,\dotsc, |U|}$. Notice that, all sets $V_1 = \set{1} \cup \compl U, V_2 = \set{2} \cup \compl U, \dotsc V_{k-1} = \set{k-1} \cup \compl U$ and $V_k = \set{k,\dotsc, |U|}$ are \onesets because none of susbsets of $U$ are \zerosets, also $\compl U$ is a \oneset so if applying $f$ to matrix created by  combining this sets as rows (first $k$ rows are sets $V_1,\dotsc V_k$ and all next are just $\compl U$)  we will result in tuple with exactly $k$ ones which is contradiction on compatibility with $\rel{{1, l}, \set{0,\dotsc, k-1, k+1, \dotsc l},l}$.\\
	See Figure~\ref{fig:gaptable} for example.
	\begin{figure}[h]
	\begin{center}
		\gaptable
	\end{center}
	\caption{Red cells are forming \onesets $V_1,\dotsc V_k$ and grey are forming \zeroset $U$}.
	\label{fig:gaptable}
	\end{figure}
\end{proof}

\begin{lemm}
	If $\clo M$ does not have a minimal \zeroset of size at least $k$ then it has a fixing set of size at most $kM$.
\end{lemm}
\begin{proof}
	Fix $f \in \clo M$. Let $\clo L$ be the largest disjoint family of minimal \zerosets of $f$ and let $U = \bigcup_{V \in \clo L} V$. Obviously $U$ intersect with every \zeroset of  $f$ so $U$ is a \onefset of $f$. Since $|\clo L| \leq M$ because of bound on antichains, and for each $V \in \clo L: |V| \leq k$ then $|U| \leq kM$ and we found a \onefset of size at most $kM$.
\end{proof}

Combining two above lemmas gives us full proof of Proposition~\ref{p:gap}. 

\subsection{Proof of Proposition~\ref{p:real}}\label{a:real}
Before we launch into the proof of the proposition we note a few easy consequences of the assumptions.
We fix $a,b,c,d,n$ satisfying the assumptions of Proposition~\ref{p:real} and denote the $C$ the set of functions in question.

\begin{lemm}
	If $f$ is in $C$ then $f$ cannot have $(n-1)c$ pairwise disjoint \zerosets i.e. is compatible with relation $\rel{(n-1)c,{2,\dotsc,(n-1)c+1},(n-1)c+1}$.
\end{lemm}
\begin{proof}
	Assume otherwise. Note that $f$ is compatible with $\rels{c,{1,\dotsc,c+1},c+1}$. Let $\sets S$ be family a family of $(n-1)c$ pairwise disjoint \zerosets and $\sets L$ be a subfamily of $S$ of size $c$. Then $\bigcup\sets L$ is a \oneset, because if we use different \zerosets from $\sets S$ as first $c$ rows, then zeros in last row does not form a minimal \zeroset (because, there are at least $c$ disjoint \zerosets covered by them), so result on last line must be equal to $1$, which implies that last row is creating a \oneset. Having $(n-1)c$ pairwise disjoint \zerosets would produce $n-1$ pairwise disjoint \onesets which is a contradiction.\\
	Second part follows from Proposition~\ref{p:fewdisjointsets}.
\end{proof}

The first order of business is to deal with the cases:
\begin{itemize}
	\item $c/d\leq 1/2<a/b$
	\item$c/d<1/2\leq a/b$.
\end{itemize}
By the discussion above the set $\flip C$ is compatible with:
\begin{itemize}
	\item $\rels{d-c,{0,\dotsc,d-1},d}$.
	\item $\rels{b-a,{1,\dotsc,b},b}$ and
	\item $\rel{1,{0,\dotsc,(n-1)c-1},(n-1)c+1}$
\end{itemize}
Therefore, by flipping the set $C$ if necessary, we can assume that $c/d<1/2\leq a/b$. Our aim is to prove, that size of smallest fixing set is bounded by $(n-2)(n-1)(c+2)$.

Take arbitrary $f\in C$ and let the arity of $f$ be $m$. If $m \leq (n-2)(n-1)(c+2)$ then, we are done.

Otherwise let $I$ be an arbitrary \oneset of $f$;
by compatibility with $\rels{a,{0,\dotsc,b-1},b}$ we conclude that $f(\compl I) = 0$ or that $\compl I$ is a minimal \oneset (by star property).\\
Now let $I$ be a minimal \oneset such that $|\compl I |\geq (n-1)c$.
Using compatibility with $\rels{c,{1,\dotsc,d},d}$ with the fact about complement of $I$ above, we conclude that if $|I|\geq (n-1)(c+2)$ we can produce $(n-1)$ disjoint \onesets which is a contradiction. Strategy for that is straightforward.\\
We will start with picking $(n-1)c$ distinct positions from $I$ denoted as $p^i_j \text{ for } 1 \leq i \leq n-1, 1 \leq j \leq c$ and another $n-1$ disjoint positions (disjoint with all $p^i_j$) from $I$ and denote them as $t^i \text{ for } 1 \leq i \leq n-1$. We will also pick the same number of positions from $\compl I$ and denote them as $q^i_j \text{ for } 1 \leq i \leq n-1, 1 \leq j \leq c$. Now we will construct $i$-th \oneset.\\
If $f(\compl I) = 0$, then set $V = \set{p^i_1, \dots, p^i_c}$ is a \oneset or $\compl V$ is a minimal \zeroset.
If $\compl I$ is a minimal \oneset, then set $V = \set{p^i_1, \dots, p^i_c, q^i_1, \dots, q^i_c}$ is a \oneset or $\compl V$ is a minimal \zeroset.\\
In both cases if $V$ is a \oneset then we are done, if $\compl V$ is a minimal \zeroset, then $V \cup \set{t^i}$ is a \oneset. See Figure~\ref{f:creatingonesets} for example of hardest case, when $d=2c+1$ and $\compl I$ is a minimal \oneset.

\begin{figure}[h]
	\begin{center}
		\creatingonesets
	\end{center}
	\caption{First 4 columns represents set $I$, first and third columns in first and second rows represents elements which are removed to make the results for first two rows equal to zero. Same thing happens in columns 9-12, but for last 2 rows. Third row creates a \oneset or minimal \zeroset - in first case we are done in second we need to add one one to get a \oneset, but columns 5-8 gives us enough space to be able to keep created \onesets disjoint.}
	\label{f:creatingonesets}
\end{figure}

Take a maximal family of disjoint, minimal \onesets of size smaller than $(n-1)(c+2)$~%
(the number of sets in the family is not greater then $n-2$)
and let $J'$ be the union of the family~%
(note that $|J'|<(n-2)(n-1)(c+2)$).
If $|J'|< (n-2)(n-1)(c+2)$  put $J$ to be any superset of $J'$ of size $(n-2)(n-1)(c+2)$, otherwise put $J=J'$.
We claim that every \oneset intersects $J$.
Indeed if $I$ is a \oneset then either $|\compl I| < (n-2)(n-1)(c+2)$ and $I\cap J\neq\emptyset$ by size considerations,
or $I$ contains a minimal \oneset of size smaller than $(n-1)(c+2)$
and it has to intersect $J$ as the family of disjoint \onesets which produced $J$ was maximal.
This immediately implies that $J$ is a \zerofset and we are done.

In the remaining case either:
\begin{itemize}
	\item $c/d<a/b<1/2$ or
	\item $1/2< c/d<a/b$
\end{itemize}
and, by flipping if necessary, we can assume we are in the first case.

Our proof is inductive and tries to reduce current case to the previous one. Notice that if we would be able to find a fixing set which will be bounded by $dM'$, where $d$ is some constant for our language and $M'$ is a boundary on a size of smallest fixing set of $\clo{M'}$ compatible with:
\begin{itemize}
	\item $\rels{a,{0,\dotsc,b-a-1},b-a}$.
	\item $\rels{c,{1,\dotsc,d-c},d-c}$ and
	\item $\rel{1,{0,\dotsc,(n-2)},n}$
\end{itemize}
then after finite number of steps we will reduce our case to first one and we will find a good boundary for $\clo M$.

Fix any $f\in C$, choose a minimal \zeroset of $f$ and denote it by $I$.
Let $f_I$ be the function of arity $|I|$ obtained by plugging zeros to the coordinates outside of $I$~%
(i.e. $f_I(J):=f(I\cap J)$).
The first goal is to show that $f_I$ has a small~%
(uniformly bounded in $C$) fixing set.
The function $f_I$ is clearly compatible with $\rels{c,{1,\dotsc,d-c},d-c}$ (because $I$ is a \zeroset so putting ones on removed columns in first $c$ rows would result in zeros on first $c$ rows of result) and with $\rel{1,{0,\dotsc,n-2},n}$.
Choose any partition of $I$ into sets of sizes greater than $a+1$ and smaller than $2a+1$,
and let $g_I$ be obtained from $f_I$ by identifying variables in each set of the partition.
Clearly $g_I$ is compatible with $\rels{c,{1,\dotsc,d-c},d-c}$ and with $\rel{1,{0,\dotsc,n-2},n}$
we claim that it is also compatible with $\rels{a,{0,\dotsc,b-a-1},b-a}$.

Suppose it is not; then there is an evaluation~%
(taking tuples of Hamming weight $a$)
which produces $b-a$ ones and on one coordinate it is not a minimal \oneset.
Let $E$ be the matrix witnessing that $g$ is not compatible with $\rels{a,{0,\dotsc,b-a-1},b-a}$~%
(i.e. $g(E)$ is a tuple of $1$'s),
and $E'$ be the same as $E$, but in one row some $1$'s became zeros~%
(and still $g(E')$ is a tuple of $1$'s).
We assume without loss of generality that $E$ and $E'$ differ on first coordinates only and evaluate
function $f$ on matrix from Figure~\ref{f:ematrix}
\begin{figure}[h]
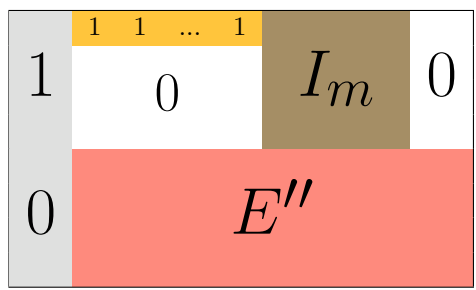

	\begin{center}
		\lasttable
	\end{center}
	\caption{Gray columns are columns not contained in \zeroset $I$, by putting ones in left upper corner we are sure, that if we add any number to the first $a$ rows, then the result will switch to 1. Yellow submatrix with ones on the first row combined with brown identity matrix next to it allows us to find out that first rows of matrix don't contain minimal \onesets.}
	\label{f:ematrix}
\end{figure}
where $E''$ is $E'$ but extended to an evaluation of $f_I$.
The result of application of $f$ to these arguments is a tuple of $1$'s
and it remains so after changing one of the arguments in the first row to $0$~%
(by minimality of the \zeroset).
This contradicts compatibility $f$ with $\rels{a,{0,\dotsc,b-1},b}$ and we conclude that
$g_I$ is compatible with $\rels{a,{0,\dotsc,b-a-1},b-a}$.

We can apply the inductive assumption and conclude that $g_I$ has a fixing set of size at most $M$~%
(and the constant $M$ is independent on the choice of $f$ and $I$).
This implies that $f_I$ has a \oneset or a \zeroset of size at most $(M+1)(2a+1)$.
Thus each minimal \zeroset $I$ has inside a small set, say $I'$, so that $I'$ is \oneset or \zeroset for $f_I$.
Let $\sets L_0$ be a family of small~%
(i.e. size bounded by $(M+1)(2a+1)$)
sets $I'$ which are \zeroset for $f_I$ defined by some minimal \zeroset $I$,
and $\sets L_1$ of these which are \oneset for some \zeroset $I$.

The family $\sets L_1$ cannot have more than $n-2$ pairwise disjoint members as each element of $\sets L_1$ is a \oneset of $f$.
Now assume that we have pairwise disjoint sets $I'_1,\dotsc,I'_{c+1}\in\sets L_0$.
Following the reasoning already explained in the sketch in Section~\ref{sect:proof},
by compatibility with $\rels{c,{1,\dotsc,d},d}$,
either $\bigcup_{i=1}^{c} I'_1$ or $\bigcup_{i=1}^{c+1}I'_i$ is a \oneset.
That means that the number of pairwise disjoint sets in $\sets L_0$ is smaller than $(c+1)(n-1)$.
Take $\sets L_1'$ to be the maximal pairwise disjoint subset of $\sets L_1$,
and $\sets L'_0$ maximal disjoint subset of $\sets L_0$.
Take $J$ to be the $\bigcup\sets L_1'\cup\bigcup\sets L_0'$;
the $|J|<(c+2)(n-1)(M+1)(2a+1)$ and every \zeroset intersects $J$ non-empty~%
(otherwise we contradict maximality)
therefore $J$ is a \onefset and we are done.


%% file: ms.bbl
\begin{thebibliography}{10}

\bibitem{TwoPlusE}
P.~{Austrin}, J.~{Håstad}, and V.~{Guruswami}.
\newblock (2 + epsilon)-sat is np-hard.
\newblock In {\em 2014 IEEE 55th Annual Symposium on Foundations of Computer
  Science}, pages 1--10, Oct 2014.
\newblock \href {http://dx.doi.org/10.1109/FOCS.2014.9}
  {\path{doi:10.1109/FOCS.2014.9}}.

\bibitem{Robust}
Libor Barto and Marcin Kozik.
\newblock Robust satisfiability of constraint satisfaction problems.
\newblock In {\em Proceedings of the Forty-fourth Annual ACM Symposium on
  Theory of Computing}, STOC '12, pages 931--940, New York, NY, USA, 2012. ACM.
\newblock URL: \url{http://doi.acm.org/10.1145/2213977.2214061}, \href
  {http://dx.doi.org/10.1145/2213977.2214061}
  {\path{doi:10.1145/2213977.2214061}}.

\bibitem{BGGraphs}
Joshua Brakensiek and Venkatesan Guruswami.
\newblock New hardness results for graph and hypergraph colorings.
\newblock In {\em Proceedings of the 31st Conference on Computational
  Complexity}, CCC '16, pages 14:1--14:27, Germany, 2016. Schloss
  Dagstuhl--Leibniz-Zentrum fuer Informatik.
\newblock URL: \url{https://doi.org/10.4230/LIPIcs.CCC.2016.14}, \href
  {http://dx.doi.org/10.4230/LIPIcs.CCC.2016.14}
  {\path{doi:10.4230/LIPIcs.CCC.2016.14}}.

\bibitem{BlendLinearRings}
Joshua Brakensiek and Venkatesan Guruswami.
\newblock An algorithmic blend of lps and ring equations for promise csps.
\newblock {\em CoRR}, abs/1807.05194, 2018.
\newblock URL: \url{http://arxiv.org/abs/1807.05194}, \href
  {http://arxiv.org/abs/1807.05194} {\path{arXiv:1807.05194}}.

\bibitem{PCSPSoda}
Joshua Brakensiek and Venkatesan Guruswami.
\newblock {\em Promise Constraint Satisfaction: Structure Theory and a
  Symmetric Boolean Dichotomy}, pages 1782--1801.
\newblock 2018.
\newblock URL:
  \url{https://epubs.siam.org/doi/abs/10.1137/1.9781611975031.117}, \href
  {http://arxiv.org/abs/https://epubs.siam.org/doi/pdf/10.1137/1.9781611975031.117}
  {\path{arXiv:https://epubs.siam.org/doi/pdf/10.1137/1.9781611975031.117}},
  \href {http://dx.doi.org/10.1137/1.9781611975031.117}
  {\path{doi:10.1137/1.9781611975031.117}}.

\bibitem{BulDich}
A.~A. Bulatov.
\newblock A dichotomy theorem for nonuniform csps.
\newblock In {\em 2017 IEEE 58th Annual Symposium on Foundations of Computer
  Science (FOCS)}, volume~00, pages 319--330, Oct. 2017.
\newblock URL: \url{doi.ieeecomputersociety.org/10.1109/FOCS.2017.37}, \href
  {http://dx.doi.org/10.1109/FOCS.2017.37} {\path{doi:10.1109/FOCS.2017.37}}.

\bibitem{BJK00}
Andrei~A. Bulatov, Andrei~A. Krokhin, and Peter Jeavons.
\newblock Constraint satisfaction problems and finite algebras.
\newblock In {\em Automata, languages and programming (Geneva, 2000)}, volume
  1853 of {\em Lecture Notes in Comput. Sci.}, pages 272--282. Springer,
  Berlin, 2000.

\bibitem{AlgebraicApproach}
Jakub Bul\'in, Andrei~A. Krokhin, and Jakub Oprsal.
\newblock Algebraic approach to promise constraint satisfaction.
\newblock {\em CoRR}, abs/1811.00970, 2018.
\newblock URL: \url{http://arxiv.org/abs/1811.00970}, \href
  {http://arxiv.org/abs/1811.00970} {\path{arXiv:1811.00970}}.

\bibitem{DinurGraphs}
Irit Dinur, Elchanan Mossel, and Oded Regev.
\newblock Conditional hardness for approximate coloring.
\newblock In {\em Proceedings of the Thirty-eighth Annual ACM Symposium on
  Theory of Computing}, STOC '06, pages 344--353, New York, NY, USA, 2006. ACM.
\newblock URL: \url{http://doi.acm.org/10.1145/1132516.1132567}, \href
  {http://dx.doi.org/10.1145/1132516.1132567}
  {\path{doi:10.1145/1132516.1132567}}.

\bibitem{FV98}
Tomás Feder and Moshe~Y. Vardi.
\newblock The computational structure of monotone monadic snp and constraint
  satisfaction: A study through datalog and group theory.
\newblock {\em SIAM Journal on Computing}, 28(1):57--104, 1998.
\newblock URL: \url{http://dx.doi.org/10.1137/S0097539794266766}, \href
  {http://arxiv.org/abs/http://dx.doi.org/10.1137/S0097539794266766}
  {\path{arXiv:http://dx.doi.org/10.1137/S0097539794266766}}, \href
  {http://dx.doi.org/10.1137/S0097539794266766}
  {\path{doi:10.1137/S0097539794266766}}.

\bibitem{Huang}
Sangxia Huang.
\newblock Improved hardness of approximating chromatic number.
\newblock In Prasad Raghavendra, Sofya Raskhodnikova, Klaus Jansen, and
  Jos{\'e} D.~P. Rolim, editors, {\em Approximation, Randomization, and
  Combinatorial Optimization. Algorithms and Techniques}, pages 233--243,
  Berlin, Heidelberg, 2013. Springer Berlin Heidelberg.

\bibitem{JCG97}
Peter Jeavons, David Cohen, and Marc Gyssens.
\newblock Closure properties of constraints.
\newblock {\em J. ACM}, 44(4):527--548, 1997.

\bibitem{VCSP}
Vladimir Kolmogorov, Andrei Krokhin, and Michal Rolinek.
\newblock The complexity of general-valued csps.
\newblock In {\em Proceedings of the 2015 IEEE 56th Annual Symposium on
  Foundations of Computer Science (FOCS)}, FOCS '15, pages 1246--1258,
  Washington, DC, USA, 2015. IEEE Computer Society.
\newblock URL: \url{http://dx.doi.org/10.1109/FOCS.2015.80}, \href
  {http://dx.doi.org/10.1109/FOCS.2015.80} {\path{doi:10.1109/FOCS.2015.80}}.

\bibitem{Sch78}
Thomas~J. Schaefer.
\newblock The complexity of satisfiability problems.
\newblock In {\em Conference {R}ecord of the {T}enth {A}nnual {ACM} {S}ymposium
  on {T}heory of {C}omputing ({S}an {D}iego, {C}alif., 1978)}, pages 216--226.
  ACM, New York, 1978.

\bibitem{ZhukDich}
D.~Zhuk.
\newblock A proof of csp dichotomy conjecture.
\newblock In {\em 2017 IEEE 58th Annual Symposium on Foundations of Computer
  Science (FOCS)}, volume~00, pages 331--342, Oct. 2017.
\newblock URL: \url{doi.ieeecomputersociety.org/10.1109/FOCS.2017.38}, \href
  {http://dx.doi.org/10.1109/FOCS.2017.38} {\path{doi:10.1109/FOCS.2017.38}}.

\end{thebibliography}
